\documentclass[11pt]{article}

\usepackage{amsmath,amssymb,amsfonts,amsthm,mathtools}
\usepackage[margin=1in]{geometry}
\usepackage{booktabs}
\usepackage{natbib}
\usepackage{enumitem}
\usepackage{array}
\usepackage{subcaption}
\usepackage{enumitem}
\usepackage{microtype} 
\usepackage{authblk}
\usepackage{bm}
\usepackage{xcolor}
\definecolor{royalblue}{rgb}{0.25, 0.41, 0.88}
\usepackage{hyperref}
\hypersetup{
    colorlinks= true,
    linkcolor= black,
    filecolor= cyan,      
    urlcolor= black, 
    citecolor= royalblue
}
\usepackage[nameinlink,capitalize,noabbrev]{cleveref}
\usepackage{tikz}
\usetikzlibrary{shapes, arrows, positioning}

\newtheorem{definition}{Definition}
\newtheorem{proposition}{Proposition}
\newtheorem{theorem}{Theorem}
\newtheorem{remark}{Remark}
\newtheorem{lemma}{Lemma}

\title{\bf Causal Discovery in Structural VAR Models Under Equal Noise Variance\vspace{5mm}}

\author[1]{SeyedSina Seyedi HasanAbadi}
\author[2]{Fahimeh Arab}
\author[1]{Erfan Nozari}
\author[3]{AmirEmad Ghassami\thanks{Correspondence Email: \texttt{ghassami@bu.edu}}}

\affil[1]{Bourns College of Engineering, University of California, Riverside}
\affil[2]{University of California, San Francisco}
\affil[3]{Department of Mathematics and Statistics, Boston University}

\date{}

\begin{document}

\maketitle

\begin{abstract}
Causal discovery from multivariate time series is challenging when causal effects may occur both across time and within the same sampling interval.  This issue is especially important in applications such as neuroscience, where the sampling rate may be coarse relative to the underlying dynamics and contemporaneous effects need not form an acyclic graph.  We study causal discovery in linear Gaussian structural VAR models under an equal noise variance assumption, meaning that the structural noise terms have a common variance.  Unlike the DAG-based cross-sectional equal noise variance setting, the time-series setting considered here does not generally yield point identification of a unique causal graph.  Instead, multiple structural VAR parameterizations can induce the same stationary observed process law.  We introduce a notion of observational equivalence tailored to this setting and show that the corresponding equivalence class is characterized by orthogonal transformations of the structural equations together with a global positive scale.  This characterization leads to an equivalence-aware model discrepancy, the observational alignment discrepancy, which compares structural models modulo transformations that preserve the observed law.  Building on this theory, we propose \textsc{ENVAR}, a sparsity-based procedure that searches over the induced observational equivalence class for a sparse normalized structural representative.  We evaluate the proposed methodology on synthetic structural VAR data and on an fMRI dataset.
\end{abstract}


\section{Introduction}
\label{sec:introduction}

Causal discovery from multivariate time series is central in applications in which interventions are expensive, slow, or ethically constrained.  In functional magnetic resonance imaging (fMRI), for example, one often observes simultaneous activity traces from many brain regions or sensors and wishes to infer which regions influence which other regions. Linear vector autoregressive (VAR) models and their structural variants are particularly powerful for answering such questions~\citep{nozari2023macroscopic}, and Granger-causal methods remain widely used in neuroscience and neuroimaging \citep{granger1969investigating,goebel2003investigating,seth2015granger}.  At the same time, fMRI data create a particular difficulty: the sampling interval is often too coarse relative to the timescale of latent neural dynamics ($\sim$1s vs. $\sim$10ms)~\citep{huettel2009functional}, and a purely lagged model can therefore misattribute within-bin interactions to lagged predictive effects. This motivates structural time-series models with both lagged and contemporaneous effects.

We study causal discovery in structural VAR models under an equal noise variance assumption, meaning that the structural noise vector is centered and has covariance equal to a common scalar multiple of the identity matrix.  The equal noise variance assumption was introduced as an identifying restriction for cross-sectional linear Gaussian structural equation models by \citet{peters2014identifiability} and was further studied by \citet{chen2019causal}.  It is especially natural when all observed variables are measured in the same physical domain and on a comparable scale. In neuroimaging, this situation arises, for example, for region-level fMRI time series, especially after standard preprocessing and normalization steps. Given the same, effective-connectivity models are routinely formulated in the fMRI literature as dynamical systems driven by homogeneous stochastic fluctuations \citep{friston2003dynamic,smith2011network,ritter2013virtual}. 

The setting considered here differs substantially from the cross-sectional structural equation model under equal noise variance of \citet{peters2014identifiability}.  In their setting, the graph is a directed acyclic graph (DAG), acyclicity provides a causal ordering, and the equal noise variance assumption yields point identification of the causal DAG from the observational Gaussian distribution. In the context of algorithm design for fMRI, however, the assumption of acyclicity is in great contrast with the highly cyclic nature of brain networks. Furthermore, in time-series data with contemporaneous effects, this point identification generally fails: multiple different graph structures can perfectly produce the exact same observed data distribution.  Because point identification is generally unavailable, the correct target is not a single structural graph, but an observational equivalence class.  Therefore, we need a careful notion of equivalence that captures exactly the information available from the observational time series.  Hence, any method that distinguishes between two observationally equivalent structures must be using assumptions beyond the observed distribution.

In this paper, we characterize the resulting observational equivalence class for structural VAR models under equal noise variance.  We show that observationally equivalent models are related by an orthogonal transformation of the structural equations together with a global positive scale.  This characterization also motivates an equivalence-aware model discrepancy, which we call the \emph{observational alignment discrepancy}.  Unlike raw Frobenius error between structural matrices, the observational alignment discrepancy assigns zero error to different representatives of the same observational law. Thus it is the appropriate notion of model-level error for both theory and simulations in this setting.

Our main contributions are as follows.
\begin{itemize}
    \item We formulate causal discovery in structural VAR(1) models under an equal noise variance assumption, allowing both lagged effects and contemporaneous effects that are not assumed to form a DAG.
    \item We introduce the relevant notion of observational equivalence for this setting and prove an exact characterization of the corresponding equivalence class.
    \item We introduce the observational alignment discrepancy as an equivalence-aware error measure for comparing structural VAR models under equal noise variance.
    \item We propose a sparsity-based estimator that searches over the observational equivalence class for a sparse normalized representative. We refer to the resulting procedure as \textsc{ENVAR}, short for Equal-Noise VAR causal discovery.
\end{itemize}

\textbf{Related work.}
Several lines of work address causal discovery from time series.  Granger causality and reduced-form VAR methods orient effects by temporal precedence but do not, by themselves, identify contemporaneous structure \citep{granger1969investigating,lutkepohl2005new,seth2015granger}.  Structural VAR methods in econometrics recover instantaneous effects only after adding identifying restrictions such as recursive orderings, sign restrictions, heteroskedasticity, or non-Gaussianity \citep{sims1980macroeconomics,moneta2011causal,hyvarinen2010estimation}.  TiMINo and related functional-model approaches use restrictions such as additive noise or independent residual processes to obtain identifiability in nonlinear time-series models \citep{peters2013causal}.  Constraint-based methods such as PCMCI and PCMCI+ learn lagged and, in later variants, contemporaneous relations under conditional-independence assumptions \citep{runge2019detecting,runge2020discovering,arab2025wholebrain}; SVAR-FCI handles latent confounding in time series but returns partial graphical information rather than a fully parameterized structural VAR \citep{malinsky2018causal}.  Score-based dynamic Bayesian-network methods such as DYNOTEARS estimate contemporaneous and lagged edges but impose an acyclicity constraint on the contemporaneous layer \citep{pamfil2020dynotears}. A broader survey of causal discovery methods for time series is given by \citet{assaad2022survey}.

\section{Model Description}
\label{sec:model_desc}

\subsection{Structural VAR Model Under Equal Noise Variance}

Let \(X_t=(X_{1t},\ldots,X_{pt})^\top\in\mathbb{R}^p\) denote the vector of observed variables at time \(t\).  We consider the linear Gaussian structural VAR(1) model
\begin{equation}\label{eq:var_gen_form}
    X_t = A_0 X_t + A_1 X_{t-1} + e_t,
    \qquad
    e_t \overset{\mathrm{i.i.d.}}{\sim} \mathcal{N}(0,\sigma^2 I_p),
    \qquad \sigma>0 .
\end{equation}
Here \(A_0\in\mathbb{R}^{p\times p}\) is the contemporaneous effect matrix and \(A_1\in\mathbb{R}^{p\times p}\) is the lag-one effect matrix.  The assumption \(\operatorname{Cov}(e_t)=\sigma^2 I_p\) is the \emph{equal noise variance} assumption.  It means that the structural noise terms are mutually independent and have a common variance.  This assumption concerns the structural noise vector \(e_t\), not the marginal variances of the observed variables \(X_t\) and not the covariance matrix of the reduced-form residuals.
We are specifically interested in models that satisfy the \emph{normalization} requirement $\operatorname{diag}(A_0)=0$.
Equivalently, if $B=I-A_0$, then \(\operatorname{diag}(B)=\mathbf{1}\).  This convention rules out instantaneous self effects (self loops) and fixes the scale of the structural equations.  Self dependence is instead represented through the lagged matrix \(A_1\).  An entry \((A_0)_{ij}\neq 0\) represents a contemporaneous effect from variable \(j\) to variable \(i\) within the same sampling interval, while an entry \((A_1)_{ij}\neq 0\) represents a lagged effect from \(X_{j,t-1}\) to \(X_{i,t}\).  The associated time-unrolled graph contains an edge \(X_{j,t}\to X_{i,t}\) when \((A_0)_{ij}\neq 0\), and an edge \(X_{j,t-1}\to X_{i,t}\) when \((A_1)_{ij}\neq 0\).

Lagged edges are naturally oriented forward in time.  Contemporaneous edges, however, may form directed cycles because they summarize interactions occurring inside one sampling interval.  This is important for data such as fMRI, where the sampling rate and hemodynamic response can be much slower than the latent neural interactions \citep{friston2003dynamic,seth2015granger}.  Accordingly, the contemporaneous matrix \(A_0\) is not assumed to be acyclic.

We require that \(B=I-A_0\) is invertible, so that the simultaneous structural equations are well defined.  Under this assumption, \cref{eq:var_gen_form} can be written in \emph{reduced form} as
\begin{equation}\label{eq:reduced_form_section2}
    X_t = \Phi X_{t-1} + u_t,
    \qquad
    \Phi := B^{-1}A_1,
    \qquad
    u_t := B^{-1}e_t .
\end{equation}
The reduced form residual covariance is
\begin{equation}\label{eq:reduced_cov_section2}
    \Sigma_u
    :=
    \operatorname{Cov}(u_t)
    =
    \sigma^2 B^{-1}B^{-\top}.
\end{equation}
Even though the structural noise vector \(e_t\) satisfies equal noise variance, the reduced-form residual covariance \(\Sigma_u\) is generally not diagonal and generally does not have equal diagonal entries.  Similarly, the stationary covariance of \(X_t\) is generally not proportional to the identity matrix.

Throughout the paper, we assume that \(A_0\), \(A_1\), and \(\sigma\) are time-invariant.  Hence the corresponding time-unrolled causal graph is fixed over time.  We also assume that the reduced-form transition matrix is stable, $\rho(\Phi)<1$, where \(\rho(\Phi)\) denotes the spectral radius of \(\Phi\).  

\textbf{Relation to Cross-Sectional Equal Noise Variance Models.}
The model in \cref{eq:var_gen_form} is related to, but substantially different from, the cross-sectional linear Gaussian structural equation model studied by \citet{peters2014identifiability} and \citet{chen2019causal}.  In the cross-sectional setting, one observes i.i.d. samples from a model of the form
\(
    X = A_0 X + e,
    ~
    e\sim \mathcal{N}(0,\sigma^2 I_p),
\)
where the directed graph associated with \(A_0\) is assumed to be a DAG.  Under the equal noise variance assumption and the DAG assumption, \citet{peters2014identifiability} show that the causal DAG is point identifiable from the observational Gaussian distribution.
Our setting differs in two essential ways.  First, the model is dynamic: the present state depends on both contemporaneous variables and the previous state.  Second, the contemporaneous effect matrix \(A_0\) is not assumed to be acyclic or triangular under any ordering.  Lagged edges are naturally oriented forward in time, but contemporaneous edges among variables at the same time point may form directed cycles.  Consequently, the contemporaneous part of the structural VAR model is not a DAG in general.
A tempting idea is to unfold a VAR(1) model over one time step and treat the pair \(Y_t := (X_{t-1}^\top, X_t^\top)^\top\) as cross-sectional data on \(2p\) variables.  This does not reduce the problem to the cross-sectional model of \citet{peters2014identifiability}.  The obstruction is twofold: after unfolding, the first time block does not have structural noise with equal noise variance, and the contemporaneous layer need not satisfy the DAG assumption required by the cross-sectional result.  See Appendix~\ref{app:compare} for details.
These observations explain why the cross-sectional point-identification result does not directly extend to the present time-series setting.  In structural VAR models under equal noise variance, different structural matrices can induce the same reduced-form transition matrix and the same reduced-form residual covariance.  Consequently, the appropriate identification target is an observational equivalence class rather than a single graph.  The next section defines this equivalence relation and characterizes it exactly.

\section{Observational Equivalence}\label{sec:obs_equiv}

In \cref{sec:model_desc}, we showed that a structural VAR model
\(
    \mathcal{M}=(A_0,A_1,\sigma),
    ~
    B=I-A_0,
\)
induces the reduced-form representation in Equations \eqref{eq:reduced_form_section2} and \eqref{eq:reduced_cov_section2}. 

\begin{definition}[Admissibility and normalization]
\label{def:admiss}
We call a structural representation \(\mathcal M=(A_0,A_1,\sigma)\) admissible
if \(B=I-A_0\) is invertible, \(\sigma>0\), and
\(
    \rho(B^{-1}A_1)<1 .
\)
It is called normalized if, in addition,
\(
    \operatorname{diag}(B)=\mathbf 1 .
\)
\end{definition}

The normalized representations are the structural VAR models introduced in
\cref{sec:model_desc}, where \(\operatorname{diag}(B)=\mathbf 1\) is equivalent to
the convention \(\operatorname{diag}(A_0)=0\).  For the equivalence
characterization, one can consider either admissible representations that are
not necessarily normalized, or the normalized equivalence class which is obtained by
intersecting the full equivalence class with the normalization constraint.

Because the reduced-form process is a stable Gaussian VAR(1), \cref{eq:reduced_form_section2} has a unique stationary Gaussian distribution with covariance \(\Sigma_X\), and the pair \((\Phi,\Sigma_u)\), together with the mean determines the full stationary law of the observed time series.  If the observed time series has a nonzero mean, an intercept can be included and removed by centering; for simplicity, we work with the mean-zero formulation.
Indeed, the stationary covariance \(\Sigma_X\) is the unique solution to the Lyapunov equation
\(
    \Sigma_X=\Phi\Sigma_X\Phi^\top+\Sigma_u,
\)
and all finite-dimensional distributions of the stationary Gaussian process are then determined by \(\Phi\) and \(\Sigma_u\).  Therefore, any two structural parameterizations that induce the same pair \((\Phi,\Sigma_u)\) are indistinguishable from observational time-series data alone.

\begin{definition}[Observational equivalence]\label{def:obs_equiv}
Let \(\mathbb{P}_{\mathcal{M}}\) denote the probability law of the stationary observed process \(\{X_t\}_{t\in\mathbb{Z}}\) generated by the structural VAR model \(\mathcal{M}\).  Two admissible models \(\mathcal{M}\) and \(\mathcal{M}'\) are \emph{observationally equivalent}, denoted
\(
    \mathcal{M}\sim_{\mathrm{obs}}\mathcal{M}',
\)
if they induce the same stationary observed process law:
\(
    \mathbb{P}_{\mathcal{M}}=\mathbb{P}_{\mathcal{M}'}.
\)
The observational equivalence class of \(\mathcal{M}\) is
\[
    \mathcal E_{\mathrm{obs}}(\mathcal M)
    :=
    \{\mathcal{M}' : \mathcal{M}'\sim_{\mathrm{obs}}\mathcal{M}\}.
\]
\end{definition}

The equivalence relation above uses all information available in the observational time series: if two models have the same stationary observed process law, then no procedure based only on the observational distribution can distinguish between them.  Conversely, if the induced reduced-form parameters differ, then the stationary Gaussian process laws differ.  Thus, observational equivalence captures exactly the information content of the observational data in this model class.

\begin{proposition}[Reduced-form characterization]\label{prop:obs_eq_reduced_form}
Let $\mathcal{M}=(A_0,A_1,\sigma)$ and $\mathcal{M}'=(A_0',A_1',\sigma')$ be two admissible structural VAR models under equal noise variance.  Define $B=I-A_0$ and $B'=I-A_0'$, and let
\(
    \Phi=B^{-1}A_1,~   
    \Sigma_u=\sigma^2B^{-1}B^{-\top},~
    \Phi'={B'}^{-1}A_1',~
    \Sigma_u'={\sigma'}^2{B'}^{-1}{B'}^{-\top}.
\)
Then $\mathcal{M}\sim_{\mathrm{obs}}\mathcal{M}'$ if and only if
\[
    \Phi=\Phi'
    \text{ and }
    \Sigma_u=\Sigma_u'.
\]
\end{proposition}
All the proofs are provided in Appendix~\ref{app:proofs}.

The next theorem gives an exact structural characterization of the equivalence class.  The key point is that equal noise variance is preserved by orthogonal transformations of the structural equations and by a common positive scaling.

\begin{theorem}[Orthogonal characterization]\label{thm:obs_eq_orbit}
Let $\mathcal{M}=(A_0,A_1,\sigma)$ and $\mathcal{M}'=(A_0',A_1',\sigma')$ be two admissible structural VAR models under equal noise variance.  Let $B=I-A_0$ and $B'=I-A_0'$.  Then $\mathcal{M}\sim_{\mathrm{obs}}\mathcal{M}'$ if and only if there exist an orthogonal matrix \(Q\in\mathcal{O}(p)=\{Q\in\mathbb{R}^{p\times p}:Q^\top Q=I_p\}\) and a scalar \(c>0\) such that
\begin{equation}
\label{eq:eq_characterization}
    B'=cQB,
    \qquad
    A_1'=cQA_1,
    \qquad
    \sigma'=c\sigma.
\end{equation}
Thus the full observational equivalence class is
\[
    \mathcal E_{\mathrm{obs}}(\mathcal M)
    =
    \left\{
    \left(I-cQB,\;cQA_1,\;c\sigma\right)
    :
    Q\in\mathcal O(p),\ c>0
    \right\}.
\]
\end{theorem}
Note that because \(\Phi=(cQB)^{-1}(cQA_1)=B^{-1}A_1\), stability is automatically preserved along the class.
The orthogonal characterization shows why the structural graph is not point identified in general.  Orthogonal transformations can change the zero pattern of \(B=I-A_0\) and \(A_1\) while preserving the same reduced-form transition matrix \(\Phi\) and the same reduced-form residual covariance \(\Sigma_u\).  Therefore, different contemporaneous and lagged graph structures can induce exactly the same observational law.  Additional assumptions, such as sparsity, are needed to select a preferred representative from the observational equivalence class.  We consider a sparsity objective in Section \ref{sec:learning}.

\begin{remark}
\label{rem:diag_normalization_orbit}
The transformation in \cref{eq:eq_characterization} is a transformation of the structural equations.
If one worked with unnormalized structural equations, the normalized observational equivalence class of \(\mathcal{M}\) is
\begin{equation}
\label{eq:norm_eq_characterization}
\mathcal E_{\mathrm{obs}}^{\mathrm{norm}}(\mathcal M)
=
\left\{
\left(I-cQB,\;cQA_1,\;c\sigma\right)
:
Q\in\mathcal O(p),\ c>0,\ 
\operatorname{diag}(cQB)=\mathbf 1
\right\}.
\end{equation}
\end{remark}

\subsection{Observational Alignment Discrepancy}
\label{subsec:alignment_discrepancy}

The characterization of equivalence class also suggests how to compare two structural models.  A naive Frobenius error such as $\lVert B-B'\rVert_F^2+\lVert A_1-A_1'\rVert_F^2$ is not appropriate for evaluating structural recovery in this setting, because two observationally equivalent representatives can have very different structural matrices.  We therefore compare models only after optimizing over the orthogonal transformations and global positive scalings that generate observational equivalence.

Let
\(
    S=[\,B\;\;A_1\,]\in\mathbb{R}^{p\times 2p},
    ~
    S'=[\,B'\;\;A_1'\,]\in\mathbb{R}^{p\times 2p}.
\)
For a tuning constant \(\eta>0\), define the \emph{observational alignment discrepancy} from \(\mathcal{M}'\) to the observational equivalence class of \(\mathcal{M}\) by
\begin{equation}
\label{eq:obs_align_discrepancy}
\Delta_{\mathrm{align}}^{\mathrm{obs}}(\mathcal M'\mid\mathcal M)
:=
\inf_{\substack{Q\in\mathcal O(p)\\ c>0}}
\left\{
\|S'-cQS\|_F^2
+
\eta(\sigma'-c\sigma)^2
\right\}.
\end{equation}
The constant \(\eta\) controls the relative weight assigned to the noise scale; one may take \(\eta=0\) if concerned only about the alignment between $S$ and $S'$.  
For admissible models, for $\eta>0$, $\Delta_{\mathrm{align}}^{\mathrm{obs}}(\mathcal{M}'\mid\mathcal{M})=0$ if and only if $\mathcal{M}'\sim_{\mathrm{obs}}\mathcal{M}$.  Thus, the discrepancy is equivalence-aware: it assigns zero discrepancy precisely to models that generate the same stationary observed process law.

\begin{remark}
\label{rem:discrepancy_not_metric}
We call \(\Delta_{\mathrm{align}}^{\mathrm{obs}}\) a discrepancy rather than a distance metric.  It is one-sided because it measures the distance from \(\mathcal{M}'\) to the equivalence class generated by \(\mathcal{M}\), and it need not satisfy symmetry or the triangle inequality.  If a symmetric reporting score is desired, one can use, for example,
\(
    \Delta_{\mathrm{sym}}(\mathcal{M},\mathcal{M}')
    =
    \frac{1}{2}
    \left\{
        \Delta_{\mathrm{align}}^{\mathrm{obs}}(\mathcal{M}'\mid\mathcal{M})
        +
        \Delta_{\mathrm{align}}^{\mathrm{obs}}(\mathcal{M}\mid\mathcal{M}')
    \right\}.
\)
\end{remark}

The discrepancy in \cref{eq:obs_align_discrepancy} has a closed-form  solution described below. 

\begin{proposition}[Closed-form]\label{prop:obs_orbit_closed_form}
Let $C=SS'^\top$, and let
\(
    C=U\operatorname{diag}(\gamma_1,\ldots,\gamma_p)V^\top
\)
be a singular value decomposition, with singular values
\(\gamma_1,\ldots,\gamma_p\geq 0\).  Define $\alpha:=\sum_{i=1}^p\gamma_i
    =
    \lVert SS'^\top\rVert_*$, where \(\lVert\cdot\rVert_*\) denotes the nuclear norm.  Then the observational alignment discrepancy is
\begin{equation}
\label{eq:obs_orbit_closed_form}
    \Delta_{\mathrm{align}}^{\mathrm{obs}}(\mathcal{M}'\mid \mathcal{M})
    =
    \lVert S'\rVert_F^2+\eta{\sigma'}^2
    -
    \frac{
        \left(\alpha+\eta\sigma\sigma'\right)^2
    }{
        \lVert S\rVert_F^2+\eta\sigma^2
    }.
\end{equation}
An optimizer is given by
\begin{equation*}
    Q^\star=VU^\top,
    \qquad
    c^\star=
    \frac{
        \alpha+\eta\sigma\sigma'
    }{
        \lVert S\rVert_F^2+\eta\sigma^2
    }.
\end{equation*}
If \(SS'^\top\) has repeated or zero singular values, the optimizer \(Q^\star\) may not be unique, but the optimal value in \cref{eq:obs_orbit_closed_form} is unique.
\end{proposition}

The closed-form expression in \cref{prop:obs_orbit_closed_form} will be used in our simulation studies in \cref{sec:results} as an equivalence-aware evaluation criterion: it measures the error of a learned structural VAR model relative to the true observational equivalence class, rather than relative to one arbitrary representative.

\begin{remark}[Normalized-class version]\label{rem:normalized_align_discrepancy}
The discrepancy in \cref{eq:obs_align_discrepancy} optimizes over the full orthogonal-scale class of the structural equations.  If one wants the literal Euclidean distance to the normalized equivalence class in \cref{eq:norm_eq_characterization}, one can impose the additional constraint $\operatorname{diag}(cQB)=\mathbf 1$ inside the infimum.  This constrained version measures distance to admissible normalized representatives, but it does not have the simple closed form in \cref{eq:obs_orbit_closed_form}.  
\end{remark}

\subsection{Scale-Free Observational Equivalence}
\label{subsec:scale_free_equiv}

In some applications, the absolute scale of the observed process is less important than the causal structure and relative dependence pattern.  We therefore also consider a coarser, \emph{scale-free} notion of equivalence.  Let \(\{X_t^{\mathcal M}\}_{t\in\mathbb Z}\) and \(\{X_t^{\mathcal M'}\}_{t\in\mathbb Z}\) denote the stationary observed processes generated by \(\mathcal M\) and \(\mathcal M'\), respectively.  We say that two admissible models are \emph{scale-free observationally equivalent}, denoted
\(
    \mathcal M\sim_{\mathrm{sf}}\mathcal M',
\)
if there exists \(a>0\) such that
\[
    \{X_t^{\mathcal M'}\}_{t\in\mathbb Z}
    \overset{d}{=}
    \{\sqrt a\,X_t^{\mathcal M}\}_{t\in\mathbb Z}.
\]
Equivalently, every finite-dimensional distribution under \(\mathcal M'\) is obtained from the corresponding finite-dimensional distribution under \(\mathcal M\) by multiplying all observed variables by the same positive constant \(\sqrt a\).  Thus, \(\sim_{\mathrm{sf}}\) ignores a common multiplicative scale of the observed time series, but still requires the same temporal dynamics and covariance structure up to a global scale factor.  For example, in neuroimaging applications, global signal scale can vary across subjects, sessions, scanners, or preprocessing pipelines, while the scientific target is often the pattern of effective connectivity rather than the absolute amplitude scale.
\begin{proposition}[Scale-free reduced-form characterization]
\label{prop:sf_reduced_form}
Let \(\mathcal M=(A_0,A_1,\sigma)\) and \(\mathcal M'=(A_0',A_1',\sigma')\) be admissible models, with reduced-form parameters \((\Phi,\Sigma_u)\) and \((\Phi',\Sigma_u')\), respectively.  Then
\(
    \mathcal M\sim_{\mathrm{sf}}\mathcal M'
\)
if and only if
\[
    \Phi=\Phi'
    ~\text{and}~
    \Sigma_u'=a\Sigma_u
\]
for some \(a>0\).
\end{proposition}
\begin{theorem}[Scale-free orthogonal characterization]
\label{thm:sf_orbit}
Let \(B=I-A_0\) and \(B'=I-A_0'\).  Then
\(
    \mathcal M\sim_{\mathrm{sf}}\mathcal M'
\)
if and only if there exist \(Q\in\mathcal O(p)\) and \(c>0\) such that
\begin{equation}
\label{eq:sf_orbit_characterization}
    B'=cQB,
    \qquad
    A_1'=cQA_1 .
\end{equation}
That is, the scale free observational equivalence class is
\[
\mathcal E_{\mathrm{sf}}(\mathcal M)
=
\left\{
\left(I-cQB,\;cQA_1,\;\tau\right)
:
Q\in\mathcal O(p),\ c>0,\ \tau>0
\right\}.
\]
\end{theorem}
The proofs of \cref{prop:sf_reduced_form,thm:sf_orbit} are the same as those of \cref{prop:obs_eq_reduced_form,thm:obs_eq_orbit}, replacing equality of \(\Sigma_u\) by proportionality of \(\Sigma_u\).

For admissible normalized representatives, this representation also satisfies
\(
    \operatorname{diag}(cQB)=\mathbf 1 .
\)
Unlike \cref{thm:obs_eq_orbit}, no condition relating \(\sigma'\) to \(c\sigma\) is imposed.  Indeed, under \cref{eq:sf_orbit_characterization},
\(
    \Sigma_u'
    =
    \frac{\tau^2}{c^2\sigma^2}\Sigma_u ,
\)
so the two models induce the same reduced-form covariance shape, but possibly with a different global scale, and we have
\[
\mathcal E_{\mathrm{sf}}^{\mathrm{norm}}(\mathcal M)
=
\left\{
\left(I-cQB,\;cQA_1,\;\tau\right)
:
Q\in\mathcal O(p),\ c>0,\ \tau>0,\ 
\operatorname{diag}(cQB)=\mathbf 1
\right\}.
\]

This coarser equivalence leads to the \emph{scale-free observational alignment discrepancy}
\begin{equation}
\label{eq:sf_align_discrepancy}
    \Delta_{\mathrm{align}}^{\mathrm{sf}}(\mathcal M'\mid\mathcal M)
:=
\inf_{\substack{Q\in\mathcal O(p)\\ c>0}}
\|S'-cQS\|_F^2,
    \qquad
    S=[\,B\;\;A_1\,],\quad S'=[\,B'\;\;A_1'\,].
\end{equation}
We have
\(
\Delta_{\mathrm{align}}^{\mathrm{obs}}(\mathcal M'\mid\mathcal M)=0
\Longleftrightarrow
\mathcal M'\in \mathcal E_{\mathrm{obs}}(\mathcal M),
\)
and
\(
\Delta_{\mathrm{align}}^{\mathrm{sf}}(\mathcal M'\mid\mathcal M)=0
\Longleftrightarrow
\mathcal M'\in \mathcal E_{\mathrm{sf}}(\mathcal M).
\)

\begin{proposition}[Closed form for the scale-free discrepancy]
\label{prop:sf_orbit_closed_form}
Let \(C=SS'^\top\), and let
\(
    C=U\operatorname{diag}(\gamma_1,\ldots,\gamma_p)V^\top
\)
be a singular value decomposition.  Define
\(
    \alpha:=\sum_{i=1}^p\gamma_i=\|SS'^\top\|_* .
\)
Then
\begin{equation*}
    \Delta_{\mathrm{align}}^{\mathrm{sf}}(\mathcal M'\mid\mathcal M)
    =
    \|S'\|_F^2
    -
    \frac{\alpha^2}{\|S\|_F^2}.
\end{equation*}
When \(\alpha>0\), an optimizer is
\[
    Q^\star=VU^\top,
    \qquad
    c^\star=\frac{\alpha}{\|S\|_F^2}.
\]
If \(\alpha=0\), the same formula gives the infimum, approached as \(c\downarrow0\).
\end{proposition}
The scale-free alignment discrepancy in \cref{eq:sf_align_discrepancy} ignores global scale mismatch, while the alignment discrepancy in \cref{eq:obs_align_discrepancy} evaluates recovery of the exact observational law.

\section{Causal Structural Learning with \textsc{ENVAR}}\label{sec:learning}

The results of \cref{sec:obs_equiv} show that the observational data identify the reduced-form parameters \((\Phi,\Sigma_u)\), but they do not generally identify a unique structural pair \((A_0,A_1)\).  The goal of this section is to use the observational equivalence characterization to construct a sparse representative of the equivalence class.  We refer to the resulting procedure as \textsc{ENVAR}, short for Equal-Noise VAR causal discovery.

The procedure has three main steps.  First, we estimate the reduced-form VAR parameters from the observed time series.  
Second, we construct a canonical, possibly unnormalized, representative of the empirical observational equivalence class.
Third, we search over this equivalence class for a sparse representative satisfying, or approximately satisfying, the diagonal normalization \(\operatorname{diag}(B)=\mathbf 1\).

\subsection{Estimating the Reduced Form}

Suppose we observe a mean-zero time series \(\{X_t\}_{t=1}^T\).  If the data have nonzero empirical mean, we first center them.  Let
\[
    n=T-1,
    \qquad
    Y=\begin{bmatrix}X_2 & X_3 & \cdots & X_T\end{bmatrix}
    \in\mathbb{R}^{p\times n},
    \qquad
    Z=\begin{bmatrix}X_1 & X_2 & \cdots & X_{T-1}\end{bmatrix}
    \in\mathbb{R}^{p\times n}.
\]
The reduced-form VAR(1) model is
\(
    X_t=\Phi X_{t-1}+u_t,
    ~
    \operatorname{Cov}(u_t)=\Sigma_u.
\)
Assuming \(ZZ^\top\) is nonsingular, the ordinary least-squares estimator is
\(
    \widehat{\Phi}
    =
    YZ^\top(ZZ^\top)^{-1}.
\)
Let $\widehat U=Y-\widehat\Phi Z$ denote the matrix of reduced-form residuals.  We estimate the reduced-form residual covariance by
\(
    \widehat\Sigma_u
    =
    \frac{1}{n}\widehat U\widehat U^\top.
\)
For the remainder of this section, we assume that \(\widehat\Sigma_u\) is positive definite.  If \(\widehat\Sigma_u\) is singular or ill-conditioned, one may replace it by a regularized estimate, for example \(\widehat\Sigma_u+\tau I_p\) with a small \(\tau>0\).

\subsection{A Canonical Representative of the Empirical Equivalence Class}

In this subsection, we construct a convenient base point for the set of structural representations that induce the empirical reduced-form parameters $(\widehat\Phi,\widehat\Sigma_u)$. 
Define the empirical precision matrix
\(
    \widehat\Omega_u:=\widehat\Sigma_u^{-1}.
\)
Let \(\widehat B_{\mathrm{can}}\) be any matrix satisfying
\begin{equation}\label{eq:B_can_def}
    \widehat B_{\mathrm{can}}^\top \widehat B_{\mathrm{can}}
    =
    \widehat\Omega_u .
\end{equation}
A natural deterministic choice is the upper-triangular Cholesky factor of \(\widehat\Omega_u\).
We then define
\begin{equation}\label{eq:Gamma_can_def}
    \widehat\Gamma_{\mathrm{can}}
    :=
    \widehat B_{\mathrm{can}}\widehat\Phi .
\end{equation}
The following proposition formalizes why
\(
    \widehat{\mathcal M}_{\mathrm{can}}
    :=
    (
        I-\widehat B_{\mathrm{can}},
        \widehat\Gamma_{\mathrm{can}},
        1
    )
\)
is a canonical unnormalized structural representation of the empirical reduced-form parameters.

\begin{proposition}
\label{prop:canonical_empirical_orbit}
Let \(\widehat\Phi\in\mathbb{R}^{p\times p}\) and \(\widehat\Sigma_u\succ0\) be given.  Let \(\widehat B_{\mathrm{can}}\) and \(\widehat\Gamma_{\mathrm{can}}\) satisfy \cref{eq:B_can_def} and \cref{eq:Gamma_can_def}, respectively.  Then the following statements hold.
\begin{enumerate}
    \item The unnormalized structural representation
    \(
        \widehat{\mathcal M}_{\mathrm{can}}
        =
        (
            I-\widehat B_{\mathrm{can}},
            \widehat\Gamma_{\mathrm{can}},
            1
        )
    \)
    induces the empirical reduced-form parameters:
    \(
        \widehat B_{\mathrm{can}}^{-1}
        \widehat\Gamma_{\mathrm{can}}
        =
        \widehat\Phi
    \text{ and }
        \widehat B_{\mathrm{can}}^{-1}
        \widehat B_{\mathrm{can}}^{-\top}
        =
        \widehat\Sigma_u .
    \)

    \item Let
    \(
        \widetilde{\mathcal M}
        =
        (\widetilde A_0,\widetilde A_1,\widetilde\sigma)
    \)
    be any structural representation with
    \(
        \widetilde B:=I-\widetilde A_0
    \)
    invertible and \(\widetilde\sigma>0\).  Then \(\widetilde{\mathcal M}\) induces the same empirical reduced-form parameters,
    \(
        \widetilde B^{-1}\widetilde A_1
        =
        \widehat\Phi,
        ~
        \widetilde\sigma^2
        \widetilde B^{-1}\widetilde B^{-\top}
        =
        \widehat\Sigma_u,
    \)
    if and only if there exist \(Q\in\mathcal O(p)\) and \(c>0\) such that
    \(
        \widetilde B
        =
        cQ\widehat B_{\mathrm{can}},
        ~
        \widetilde A_1
        =
        cQ\widehat\Gamma_{\mathrm{can}},
        ~
        \widetilde\sigma
        =
        c .
    \)
\end{enumerate}
\end{proposition}

By \cref{prop:canonical_empirical_orbit}, the full unnormalized empirical observational equivalence class generated by the canonical representative is
\begin{equation}
 \label{eq:empirical_equiv_class}
    \widehat{\mathcal E}_{\mathrm{obs}}
    :=
    \left\{
    \left(
        I-cQ\widehat B_{\mathrm{can}},
        \;
        cQ\widehat\Gamma_{\mathrm{can}},
        \;
        c
    \right)
    :
    Q\in\mathcal O(p),\ c>0
    \right\}.
\end{equation}
Equivalently, writing the structural representation in terms of \(B=I-A_0\), every member of this class has
\begin{equation*}
    \widehat B(Q,c)
    :=
    cQ\widehat B_{\mathrm{can}},
    \qquad
    \widehat A_1(Q,c)
    :=
    cQ\widehat\Gamma_{\mathrm{can}},
    \qquad
    \widehat\sigma(Q,c)
    :=
    c .
\end{equation*}
All triples in \cref{eq:empirical_equiv_class} induce the same empirical reduced-form parameters \((\widehat\Phi,\widehat\Sigma_u)\).
The normalized empirical representatives, if they exist, are the members satisfying
\(
    \operatorname{diag}\!\left(cQ\widehat B_{\mathrm{can}}\right)=\mathbf 1 .
\)

\begin{remark}
The matrix \(\widehat B_{\mathrm{can}}\) is not unique.  Any two matrices satisfying
\(
    B_1^\top B_1
    =
    B_2^\top B_2
    =
    \widehat\Sigma_u^{-1}
\)
differ by a left orthogonal transformation.  Therefore, choosing the upper-triangular Cholesky factor simply fixes a convenient base point for the same empirical equivalence class. 
\end{remark}

\subsection{Sparse Representative Selection}

Observational equivalence alone identifies an equivalence class, not a unique graph.  To select a representative, we use sparsity as an additional modeling principle.  The diagonal entries of \(A_0\) are controlled by the normalization \(\operatorname{diag}(A_0)=0\), equivalently \(\operatorname{diag}(B)=\mathbf 1\), and are therefore not included in the sparsity penalty.  We penalize the off-diagonal entries of \(A_0\), which encode contemporaneous effects between distinct variables, and all entries of \(A_1\), which encode lagged effects.
Let \(\operatorname{offdiag}(M)\) be the matrix obtained from \(M\) by setting its diagonal entries to zero, and let \(\|M\|_1=\sum_{i,j}|M_{ij}|\) denote the entrywise \(\ell_1\) norm.  The ideal constrained sparse representative is obtained by solving
\begin{equation}\label{eq:envar_constrained}
\begin{aligned}
    (\widehat Q,\widehat c)
    \in
    \arg\min_{\substack{Q\in\mathcal O(p)\\ c>0}}
    \quad
    &
    \lambda_0
    \left\|
        \operatorname{offdiag}\!\left(I-cQ\widehat B_{\mathrm{can}}\right)
    \right\|_1
    +
    \lambda_1
    \left\|
        cQ\widehat\Gamma_{\mathrm{can}}
    \right\|_1
    \\
    \text{subject to}
    \quad
    &
    \operatorname{diag}\!\left(cQ\widehat B_{\mathrm{can}}\right)
    =
    \mathbf 1 .
\end{aligned}
\end{equation}
Here \(\lambda_0,\lambda_1\geq 0\) control the relative preference for sparse contemporaneous and lagged effects.

\begin{remark}
The diagonal constraint in \cref{eq:envar_constrained} is essential.  Without this constraint, taking \(c\downarrow0\) drives both sparsity penalties to zero.  Thus the scale \(c\) cannot be optimized without a normalization constraint.  Conversely, fixing \(c=1\) is generally not appropriate either, because the Cholesky representative \(\widehat B_{\mathrm{can}}\) has arbitrary global scale relative to the normalized structural equations.
\end{remark}

For numerical work, we use a soft version of \cref{eq:envar_constrained}.  To avoid boundary shrinkage of the global scale, we solve the penalized problem over a wide compact interval
\(
    \mathcal C=[c_{\min},c_{\max}],
    ~
    0<c_{\min}<c_{\max}<\infty .
\)
The penalized problem is
\begin{equation}\label{eq:envar_penalized}
\begin{aligned}
    (\widehat Q,\widehat c)
    \in
    \arg\min_{\substack{Q\in\mathcal O(p)\\ c\in\mathcal{C}}}
    \quad
    \!\!\!\!\!\!\!
    \lambda_0
    \left\|
        \operatorname{offdiag}\!\left(I\!-\!cQ\widehat B_{\mathrm{can}}\right)
    \right\|_1
    \!\!+\!
    \lambda_1
    \left\|
        cQ\widehat\Gamma_{\mathrm{can}}
    \right\|_1
    \!\!+\!
    \frac{\mu}{2}
    \left\|
        \operatorname{diag}\!\left(cQ\widehat B_{\mathrm{can}}\right)
        \!-\!
        \mathbf 1
    \right\|_2^2
    .
\end{aligned}
\end{equation}
The parameter \(\mu>0\) controls the strength of the diagonal-normalization penalty.  Larger values of \(\mu\) force the solution closer to the normalized model class.

After solving either \cref{eq:envar_constrained} or \cref{eq:envar_penalized}, \textsc{ENVAR} returns
\begin{equation*}
    \widehat B_{\textsc{ENVAR}}
    =
    \widehat c\,\widehat Q\,\widehat B_{\mathrm{can}},
    \qquad
    \widehat A_{0,\textsc{ENVAR}}
    =
    I-\widehat B_{\textsc{ENVAR}},
\qquad
    \widehat A_{1,\textsc{ENVAR}}
    =
    \widehat c\,\widehat Q\,\widehat\Gamma_{\mathrm{can}},
    \qquad
    \widehat\sigma_{\textsc{ENVAR}}
    =
    \widehat c.
\end{equation*}

\section{Simulation Study and Application}\label{sec:results}

\begin{figure}[t]
  \centering
  \begin{subfigure}[b]{0.40\textwidth}
    \centering
    \includegraphics[height=5.00cm]{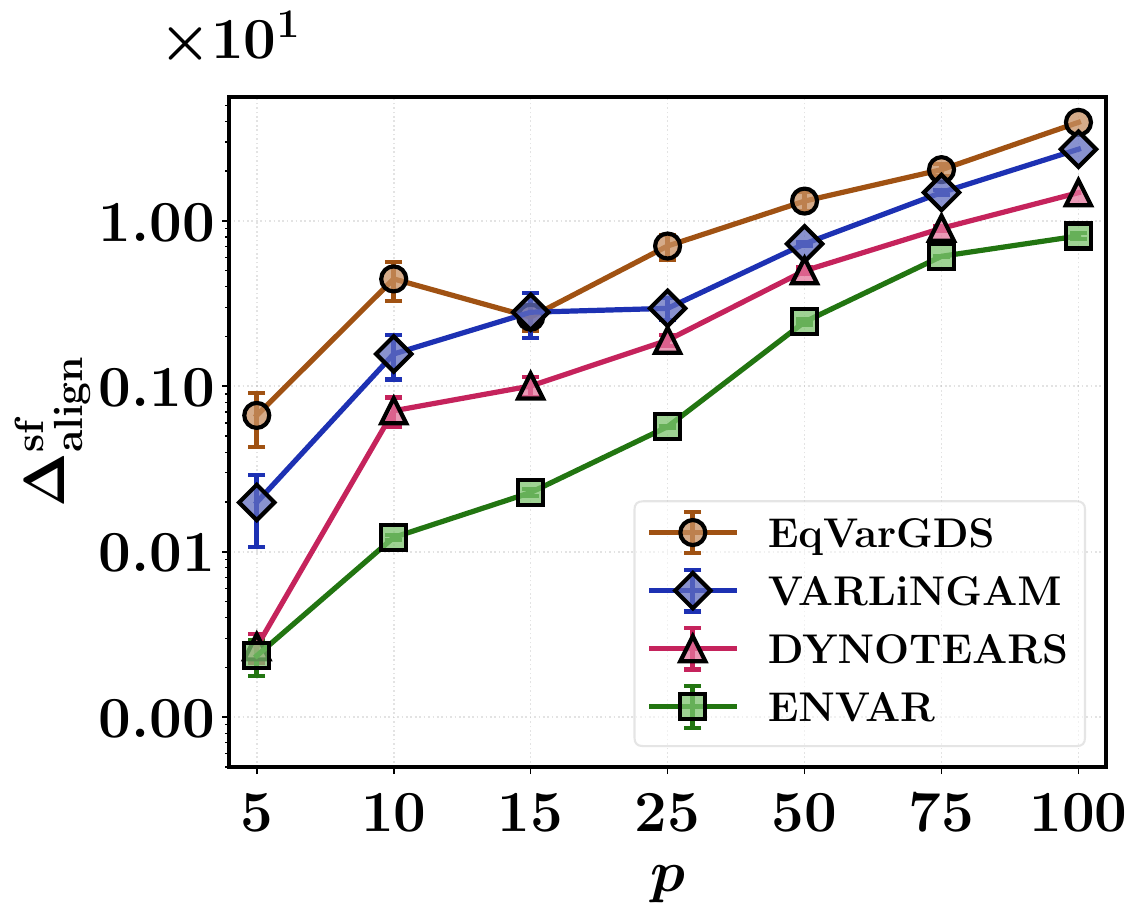}
    \caption{scale-free alignment discrepancy}
    \label{fig:obs_orb_disc_eq_var}
  \end{subfigure}
  \hfill
  \begin{subfigure}[b]{0.55\textwidth}
    \centering
    \includegraphics[height=5.00cm]{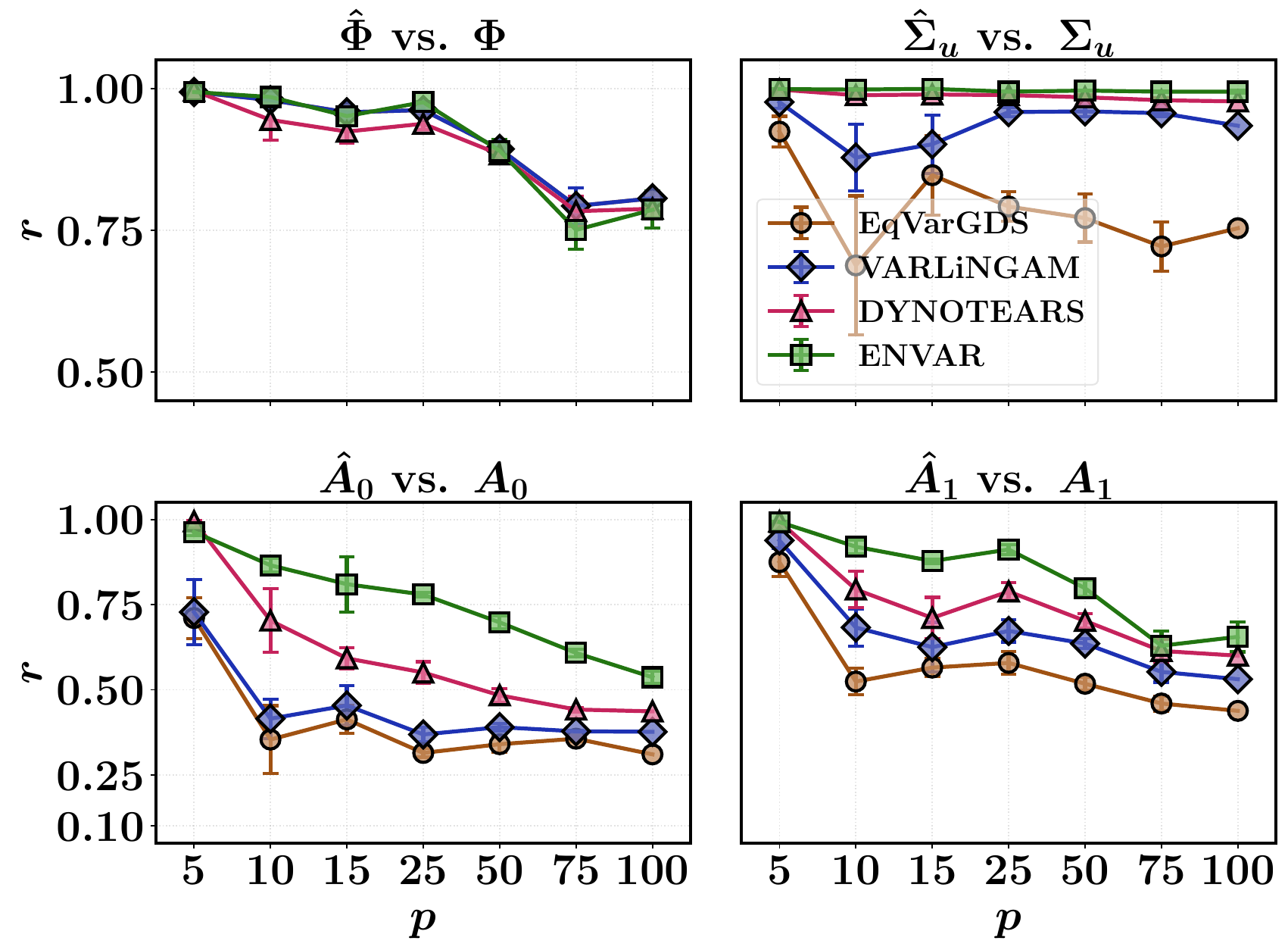}
    \caption{Pearson correlation coefficient ($r$)}
    \label{fig:corr_eq_var}
  \end{subfigure}
  \caption{Performance comparison of EqVarGDS, ENVAR, VARLiNGAM, and DYNOTEARS vs. the number of nodes, under the equal variance assumption for the noise. The number of samples is considered $T = 1000$. The error bars represent SEM across 5 episodes. \textbf{(a)} The scale-free alignment discrepancy of the predicted lag matrices to the ground truth equivalence class. Note the logarithmic scale on the ordinate. \textbf{(b)} The Pearson correlation coefficient ($r$) between the predicted reduced form parameters and their corresponding ground truths. Only significant correlations with p-value $<0.05$ are considered for better comparison. EqVarGDS and VARLiNGAM lines overlap with ENVAR in the top panel.}
  \label{fig:combined_eq_var_results}
\end{figure}

In this section, we evaluate ENVAR in two settings: a synthetic VAR(1) benchmark with known ground-truth and real fMRI data from the Human Connectome Project (HCP)~\citep{van2013wu}.

\subsection{Causal Discovery from Synthetic Data}

To evaluate ENVAR against state-of-the-art alternatives, we generated a synthetic dataset based on the VAR(1) model described in \cref{sec:model_desc} and Erdős–Rényi topology. Specifically, the binary structure (support) of $A_0$ (enforcing $(A_0)_{ii} = 0$) and $A_1$ were generated as independent Erdős–Rényi random graphs with 0.3 edge probability, for each.  Non-zero edge weights were then sampled from the uniform distribution $\mathcal{U}(-1.0, 1.0)$ and $A_1$ was subsequently scaled to achieve the desired spectral radius for $\Phi$ (and ensure stability of~\cref{eq:reduced_form_section2}).  The data matrix $X \in \mathbb{R}^{p \times T}$ was then sequentially generated via~\cref{eq:reduced_form_section2}, where $u_t = (I - A_0)^{-1} e_t$ and $e_t \overset{i.i.d.}{\sim} \mathcal{N}(0, \sigma^2 I_p)$.

We compare the accuracy of learned graphs by ENVAR against three baselines: the Equal Variance Greedy DAG Search method (EqVarGDS) \citep{peters2014identifiability}, VARLiNGAM \citep{hyvarinen2010estimation}, and DYNOTEARS \citep{pamfil2020dynotears}. 
\cref{fig:obs_orb_disc_eq_var} shows the scale-free observational alignment discrepancy ($\Delta_{\mathrm{align}}^{\mathrm{sf}}$, cf.~\cref{eq:sf_align_discrepancy}) as a function of the number of nodes $p$. $\Delta_{\mathrm{align}}^{\mathrm{sf}}$ increases with number of nodes for all methods, but ENVAR achieves the lowest discrepancy across almost the entire range. As a secondary metric, \cref{fig:corr_eq_var} further reports the Pearson correlation coefficient between the predicted and ground-truth reduced-form (i.e., identifiable) parameters, and the structural lag matrices. It can be seen that ENVAR attains the highest correlation in all settings. This shows that ENVAR not only succeeds in its objective of recovering models that are closest to the correct observational equivalence class, but is also able to learn the fundamentally-unidentifiable matrices $A_0$ and $A_1$ with greater accuracy than state of the art.  
In contrast, e.g., DYNOTEARS (the second-best method in most metrics) strictly constrains the contemporaneous lag matrix to be a DAG without accounting for the observational equivalence class of the ground truth, resulting in potentially sacrificing overall predictive accuracy ($\Phi$ and $\Sigma_u$) in favor of structural matrix estimation.

Furthermore, to assess the robustness of ENVAR to violations in its equal-variance assumption, we conducted a sensitivity analysis in which the standard deviations of the per-region noise distributions were themselves drawn independently as $\sigma_i \sim \mathcal{N}(\sigma_{\mathrm{nom}}, \sigma_{\mathrm{std}}^2)$, so that $\sigma_{\mathrm{std}} = 0$ recovers the equal-variance setting and larger $\sigma_{\mathrm{std}}$ corresponds to greater violations. We set $\sigma_{\mathrm{nom}} = 1$ and varied $\sigma_{\mathrm{std}}$ and $p$. The resulting observational alignment discrepancy and the Pearson correlations between predicted and ground-truth reduced-form parameters are shown in \cref{fig:supp_obs_orb_disc_eq_var_violation_ENVAR} and \cref{fig:supp_corr_eq_var_violation_ENVAR}, respectively. Notably, ENVAR's performance degrades very gracefully as $\sigma_{\mathrm{std}}$ increases, particularly for smaller graph sizes. Thus, together, these results demonstrate ENVAR's robust and significantly higher accuracy compared to the state of the art when noise variances are the same or even close to each other among all nodes (\cref{fig:supp_corr_eq_var_violation}). It can be seen that even under the violation of the equal variance assumption, still ENVAR is performing better than other methods.

\subsection{Causal Discovery from Real fMRI Data}

We next evaluated ENVAR on fMRI data from Human Connectome Project (HCP) \citep{van2013wu}. We used motor task data from 25 subjects, motivated by the significantly greater background knowledge about the human motor control system that can be used to evaluate outcomes. Data was preprocessed according to standard practice (see~\cref{sec:exp-details} for details). ENVAR was then applied to the data from each subject and the resulting causal conductivities were evaluated on the basis of how well nodal centralities align with known activation maps during human motor control (see~\cref{sec:exp-details} for details).

\begin{figure}
    \centering
    \includegraphics[width=0.8\linewidth]{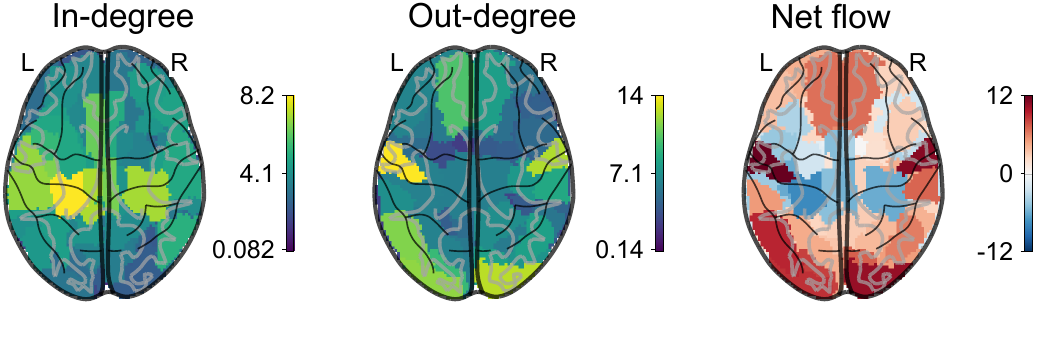} 
    \caption{Nodal centralities of the binarized ENVAR graph during the HCP
motor task. For each cortical parcel: (left) in-degree, the
number of incoming edges; (center) out-degree, outgoing edges;
(right) net flow, out-degree minus in-degree (red = source,
blue = sink). Per subject and run, $A_0$ and $A_1$ were thresholded and
binarized, edges counted across both lags, and degrees averaged across subjects and runs. Subcortical parcels are included in the
computation but omitted from the figure.
}
    \label{fig:glass_brain}
\end{figure}

As seen in \cref{fig:glass_brain}, the resulting nodal centralities align remarkably well with expected activations in the motor cortex. This alignment is more clearly visible from the NeuroSynth word cloud in \cref{fig:supp_hcp_task_motor_wc}, demonstrating the list of words most strongly associated with nodal centralities of graphs learned by ENVAR based on the neuroscience literature~\citep{yarkoni2011large}. These results therefore clearly demonstrate the power of ENVAR in extracting meaningful graphs from highly noisy and temporally under-sampled fMRI data, and motivate future applications in areas where ground-truth causal connections are more poorly understood.

\section{Conclusion}\label{sec:conclusion}

We studied causal discovery in structural VAR models under an equal noise variance assumption, allowing contemporaneous effects that need not form a DAG.  We showed that, unlike the cross-sectional equal noise variance setting, the structural graph is generally not point identified from observational time-series data.  The appropriate target is therefore an observational equivalence class, which we characterized through orthogonal transformations of the structural equations and a global positive scale.  This characterization led to the observational alignment discrepancy and to \textsc{ENVAR}, a sparsity-based procedure for selecting a normalized structural representative. In simulations, \textsc{ENVAR} achieved lower observational alignment discrepancy than the competing methods considered. In the fMRI application, the cognitive terms emphasized by \textsc{ENVAR} were more closely aligned with the motor-task structure, suggesting that the method can produce interpretable effective-connectivity summaries in neuroimaging data.

\textbf{Limitations and future work.} The proposed method relies on a linear first-order Markov model, lack of unobserved confounding, and stationarity. Also notable is the need for and use of nonconvex optimization solvers for ENVAR that may lead to identification of suboptimal local minima. Future work is needed to address these limitations, including, in particular, addressing nonstationarities that arise in fMRI data and learning higher-order models that can capture slower dynamics.

\bibliographystyle{plainnat}
\bibliography{references.bib}

\appendix

\begin{center}

\newpage{\Large \bf Appendix for\\``Causal Discovery in Structural VAR Models Under\\ Equal Noise Variance''}

\vspace{5mm}

SeyedSina Seyedi HasanAbadi, Fahimeh Arab, Erfan Nozari, and AmirEmad Ghassami

\vspace{10mm}

\end{center}

\section{Relation to Cross-Sectional Equal Noise Variance Models}
\label{app:compare}

The model in \cref{eq:var_gen_form} is related to, but substantially different from, the cross-sectional linear Gaussian structural equation model studied by \citet{peters2014identifiability} (and \citet{chen2019causal}).  In the cross-sectional setting, one observes i.i.d. samples from a model of the form
\(
    X = A_0 X + e,
    ~
    e\sim \mathcal{N}(0,\sigma^2 I_p),
\)
where the directed graph associated with \(A_0\) is assumed to be a DAG.  Under the equal noise variance assumption and the DAG assumption, \citet{peters2014identifiability} show that the causal DAG is point identifiable from the observational Gaussian distribution.

Our setting differs in two essential ways.  First, the model is dynamic: the present state depends on both contemporaneous variables and the previous state.  Second, the contemporaneous effect matrix \(A_0\) is not assumed to be acyclic or triangular under any ordering.  The latter distinction is particularly important.  In a time-unrolled graph, lagged edges point from time \(t-1\) to time \(t\) and therefore do not themselves create directed cycles across time.  However, contemporaneous edges among variables at the same time can form directed cycles.  Consequently, the contemporaneous part of the structural VAR model is not a DAG in general.

A tempting idea may be to unfold a VAR(1) model over one time step and treat the extended vector $Y_t := (X_{t-1}^\top, X_t^\top)^\top$ as cross-sectional data on \(2p\) variables.  This does not reduce the problem to the cross-sectional model of \citet{peters2014identifiability}, however.  To see why, consider the first block \(X_{t-1}\) and the second block \(X_t\) of the unfolded vector.  The variables at time \(t\) are driven by the structural noise \(e_t\).  By contrast, the variables at time \(t-1\) already contain the accumulated effect of earlier structural noises and earlier causal interactions.  Therefore, the effective noise attached to the first \(p\) unfolded variables is not another copy of \(e_t\).

This can be made explicit by applying the structural equation one step earlier:
\[
Y_t = \begin{bmatrix}
    A_0 & 0 \\
    A_1 & A_0
\end{bmatrix} Y_t + \begin{bmatrix}
    A_1 X_{t-2} + e_{t-1} \\
    e_t
\end{bmatrix}
\]
If the unobserved variable \(X_{t-2}\) is marginalized out, then the effective residual for the first block is $A_1 X_{t-2} + e_{t-1}$. Under stationarity, its covariance is
\(
    \operatorname{Cov}(A_1 X_{t-2}+e_{t-1})
    =
    A_1\Sigma_XA_1^\top+\sigma^2 I_p,
\)
which is generally not proportional to \(I_p\) and may contain cross-variable correlations.  Hence the unfolded \(2p\)-dimensional cross-sectional representation does not satisfy the equal noise variance condition required by the cross-sectional identifiability result.

There is also a graph-related obstruction.  The cross-sectional result of \citet{peters2014identifiability} relies on the existence of a DAG ordering.  In the structural VAR model considered here, the contemporaneous matrix \(A_0\) may contain cycles.  Therefore, even apart from the failure of equal noise variance after unfolding, the unfolded model does not generally belong to the DAG-based model class for which point identification is known.

These observations explain why the cross-sectional point-identification result does not directly extend to the present time-series setting.  In structural VAR models under equal noise variance, different structural matrices can induce the same reduced-form transition matrix and the same reduced-form residual covariance.  Consequently, the appropriate identification target is an observational equivalence class rather than a single graph. 

\section{Proofs}
\label{app:proofs}

We collect the proofs of the theoretical results in Sections~\ref{sec:obs_equiv} and~\ref{sec:learning}.  Throughout, all models are assumed to satisfy the admissibility conditions stated in Definition~\ref{def:admiss}.  In particular, \(B=I-A_0\) is invertible, \(\sigma>0\), and the reduced-form transition matrix \(\Phi=B^{-1}A_1\) is stable.

\begin{lemma}[Gram--orthogonal factorization]
\label{lem:gram_orthogonal}
Let \(C,D\in\mathbb R^{p\times p}\) be invertible matrices.  Suppose there exists \(\lambda>0\) such that
\[
    D^\top D=\lambda C^\top C .
\]
Then there exists an orthogonal matrix \(Q\in\mathcal O(p)\) such that
\[
    D=\sqrt{\lambda}\,QC .
\]
Equivalently, if \(D^\top D=\beta^2 C^\top C\) for some \(\beta>0\), then there exists \(Q\in\mathcal O(p)\) such that
\[
    D=\beta QC .
\]
\end{lemma}

\begin{proof}
Define
\[
    Q:=\lambda^{-1/2}DC^{-1}.
\]
We show that \(Q\) is orthogonal.  Since \(C\) is invertible,
\[
\begin{aligned}
    Q^\top Q
    &=
    \lambda^{-1}C^{-\top}D^\top DC^{-1}  \\
    &=
    \lambda^{-1}C^{-\top}(\lambda C^\top C)C^{-1} \\
    &=
    C^{-\top}C^\top CC^{-1}
    =
    I_p .
\end{aligned}
\]
Hence \(Q\in\mathcal O(p)\).  By the definition of \(Q\),
\[
    D=\sqrt{\lambda}\,QC.
\]
The equivalent statement follows by taking \(\lambda=\beta^2\).
\end{proof}

\subsection{Proof of Proposition~\ref{prop:obs_eq_reduced_form}}

\begin{proof}
Recall that the reduced-form representation induced by \(\mathcal M\) is
\[
    X_t=\Phi X_{t-1}+u_t,
    \qquad
    \operatorname{Cov}(u_t)=\Sigma_u,
\]
where
\[
    \Phi=B^{-1}A_1,
    \qquad
    \Sigma_u=\sigma^2B^{-1}B^{-\top}.
\]
Since the model is admissible, \(\rho(\Phi)<1\), and therefore the reduced-form VAR(1) process has a unique stationary Gaussian law.  Its stationary covariance \(\Sigma_X\) is the unique solution of the Lyapunov equation
\[
    \Sigma_X=\Phi\Sigma_X\Phi^\top+\Sigma_u.
\]

First suppose that
\[
    \Phi=\Phi',
    \qquad
    \Sigma_u=\Sigma_u'.
\]
Then the two models induce the same stable reduced-form VAR(1) process.  They therefore have the same stationary covariance, the same transition law, and hence the same finite-dimensional Gaussian distributions.  Thus,
\[
    \mathbb P_{\mathcal M}=\mathbb P_{\mathcal M'},
\]
so \(\mathcal M\sim_{\mathrm{obs}}\mathcal M'\).

Conversely, suppose that
\[
    \mathcal M\sim_{\mathrm{obs}}\mathcal M',
\]
so that the two stationary observed processes have the same law.  Then their lag-zero and lag-one covariance matrices are equal.  Let
\[
    \Gamma_0:=\operatorname{Cov}(X_t),
    \qquad
    \Gamma_1:=\operatorname{Cov}(X_t,X_{t-1}).
\]
For a stationary VAR(1) process,
\[
    \Gamma_1
    =
    \operatorname{Cov}(\Phi X_{t-1}+u_t,X_{t-1})
    =
    \Phi\Gamma_0,
\]
because \(u_t\) is independent of the past.  Moreover, \(\Gamma_0\) is positive definite since \(\Sigma_u\) is positive definite.  Hence
\[
    \Phi=\Gamma_1\Gamma_0^{-1}.
\]
Since the two process laws are equal, they have the same \(\Gamma_0\) and \(\Gamma_1\), and therefore they have the same reduced-form transition matrix:
\[
    \Phi=\Phi'.
\]
Finally, using the Lyapunov equation,
\[
    \Sigma_u
    =
    \Gamma_0-\Phi\Gamma_0\Phi^\top.
\]
Since \(\Gamma_0\) and \(\Phi\) agree for the two processes, it follows that
\[
    \Sigma_u=\Sigma_u'.
\]
This proves the desired equivalence.
\end{proof}

\subsection{Proof of Theorem~\ref{thm:obs_eq_orbit}}

\begin{proof}
We prove both directions.

First suppose that
\[
    \mathcal M\sim_{\mathrm{obs}}\mathcal M'.
\]
By Proposition~\ref{prop:obs_eq_reduced_form}, the two models induce the same reduced-form parameters:
\[
    \Phi=\Phi',
    \qquad
    \Sigma_u=\Sigma_u'.
\]
Equality of the reduced-form residual covariance matrices gives
\[
    \sigma^2B^{-1}B^{-\top}
    =
    {\sigma'}^2{B'}^{-1}{B'}^{-\top}.
\]
Taking inverses on both sides yields
\[
    \sigma^{-2}B^\top B
    =
    {\sigma'}^{-2}{B'}^\top B'.
\]
Equivalently,
\[
    {B'}^\top B'
    =
    \left(\frac{\sigma'}{\sigma}\right)^2B^\top B.
\]
Define
\[
    c:=\frac{\sigma'}{\sigma}>0.
\]
By Lemma~\ref{lem:gram_orthogonal}, applied with \(C=B\), \(D=B'\), and \(\lambda=c^2\), there exists \(Q\in\mathcal O(p)\) such that
\[
    B'=cQB.
\]
The equality \(\Phi=\Phi'\) gives
\[
    B^{-1}A_1={B'}^{-1}A_1'.
\]
Multiplying by \(B'\) on the left gives
\[
    A_1'=B'B^{-1}A_1.
\]
Using \(B'=cQB\), we obtain
\[
    A_1'=cQBB^{-1}A_1=cQA_1.
\]
Finally, by the definition of \(c\),
\[
    \sigma'=c\sigma.
\]
Thus,
\[
    B'=cQB,
    \qquad
    A_1'=cQA_1,
    \qquad
    \sigma'=c\sigma.
\]

Conversely, suppose there exist \(Q\in\mathcal O(p)\) and \(c>0\) such that
\[
    B'=cQB,
    \qquad
    A_1'=cQA_1,
    \qquad
    \sigma'=c\sigma.
\]
Then the reduced-form transition matrix of \(\mathcal M'\) is
\[
\begin{aligned}
    \Phi'
    &=
    {B'}^{-1}A_1'   \\
    &=
    (cQB)^{-1}(cQA_1) \\
    &=
    B^{-1}Q^\top Q A_1 \\
    &=
    B^{-1}A_1
    =
    \Phi.
\end{aligned}
\]
Similarly, the reduced-form residual covariance is
\[
\begin{aligned}
    \Sigma_u'
    &=
    {\sigma'}^2{B'}^{-1}{B'}^{-\top}  \\
    &=
    c^2\sigma^2(cQB)^{-1}(cQB)^{-\top} \\
    &=
    c^2\sigma^2
    \left(c^{-1}B^{-1}Q^\top\right)
    \left(c^{-1}QB^{-\top}\right) \\
    &=
    \sigma^2B^{-1}Q^\top Q B^{-\top} \\
    &=
    \sigma^2B^{-1}B^{-\top}
    =
    \Sigma_u.
\end{aligned}
\]
Therefore \(\Phi'=\Phi\) and \(\Sigma_u'=\Sigma_u\).  By Proposition~\ref{prop:obs_eq_reduced_form},
\[
    \mathcal M\sim_{\mathrm{obs}}\mathcal M'.
\]
The displayed form of the observational equivalence class follows by writing \(A_0'=I-B'=I-cQB\).
\end{proof}

\subsection{Proof of Proposition~\ref{prop:obs_orbit_closed_form}}

\begin{proof}
Let
\[
    S=[\,B\;\;A_1\,],
    \qquad
    S'=[\,B'\;\;A_1'\,].
\]
The observational alignment discrepancy is defined as
\[
\Delta_{\mathrm{align}}^{\mathrm{obs}}(\mathcal M'\mid\mathcal M)
=
\inf_{\substack{Q\in\mathcal O(p)\\ c>0}}
\left\{
    \|S'-cQS\|_F^2
    +
    \eta(\sigma'-c\sigma)^2
\right\}.
\]
We expand the objective for fixed \(Q\) and \(c\).  Using
\[
    \|X-Y\|_F^2=\|X\|_F^2+\|Y\|_F^2-2\operatorname{tr}(YX^\top),
\]
we get
\[
    \|S'-cQS\|_F^2
    =
    \|S'\|_F^2
    +
    c^2\|S\|_F^2
    -
    2c\,\operatorname{tr}(QS{S'}^\top).
\]
Also,
\[
    \eta(\sigma'-c\sigma)^2
    =
    \eta{\sigma'}^2
    +
    \eta c^2\sigma^2
    -
    2\eta c\sigma\sigma'.
\]
Therefore the objective is
\[
\begin{aligned}
    L(Q,c)
    &=
    \|S'\|_F^2+\eta{\sigma'}^2
    +
    c^2(\|S\|_F^2+\eta\sigma^2)    \\
    &\qquad
    -
    2c\left\{
        \operatorname{tr}(QS{S'}^\top)
        +
        \eta\sigma\sigma'
    \right\}.
\end{aligned}
\]

For fixed \(c>0\), minimizing \(L(Q,c)\) over \(Q\in\mathcal O(p)\) is equivalent to maximizing
\[
    \operatorname{tr}(QS{S'}^\top).
\]
Let
\[
    C=S{S'}^\top
\]
and let
\[
    C=U\operatorname{diag}(\gamma_1,\ldots,\gamma_p)V^\top
\]
be a singular value decomposition.  By von Neumann's trace inequality,
\[
    \max_{Q\in\mathcal O(p)}
    \operatorname{tr}(QC)
    =
    \sum_{i=1}^p\gamma_i
    =
    \|C\|_*.
\]
With the notation
\[
    \alpha:=\sum_{i=1}^p\gamma_i=\|S{S'}^\top\|_*,
\]
one maximizer is
\[
    Q^\star=VU^\top.
\]
Indeed, for this choice,
\[
    V^\top Q^\star U
    =
    V^\top VU^\top U
    =
    I_p,
\]
which attains equality in the trace inequality.

Substituting the optimal value \(\alpha\) for the trace term reduces the problem to a one-dimensional minimization over \(c>0\):
\[
    \|S'\|_F^2+\eta{\sigma'}^2
    +
    c^2(\|S\|_F^2+\eta\sigma^2)
    -
    2c(\alpha+\eta\sigma\sigma').
\]
This is a convex quadratic in \(c\).  Since \(B\) is invertible, \(S\neq0\), and hence \(\|S\|_F^2>0\).  For \(\eta>0\), the coefficient \(\alpha+\eta\sigma\sigma'\) is strictly positive, so the minimizer over \(c>0\) is attained at
\[
    c^\star
    =
    \frac{\alpha+\eta\sigma\sigma'}{\|S\|_F^2+\eta\sigma^2}.
\]
Substituting this value gives
\[
    \Delta_{\mathrm{align}}^{\mathrm{obs}}(\mathcal M'\mid\mathcal M)
    =
    \|S'\|_F^2+\eta{\sigma'}^2
    -
    \frac{
        \left(\alpha+\eta\sigma\sigma'\right)^2
    }{
        \|S\|_F^2+\eta\sigma^2
    }.
\]
This is the claimed closed-form expression.  If \(SS'^\top\) has repeated or zero singular values, the maximizing orthogonal matrix need not be unique, but the maximum trace value \(\alpha\), and hence the optimal discrepancy value, is unique.
\end{proof}

\subsection{Proof of Proposition~\ref{prop:sf_reduced_form}}

\begin{proof}
Let \(\{X_t^{\mathcal M}\}_{t\in\mathbb Z}\) and \(\{X_t^{\mathcal M'}\}_{t\in\mathbb Z}\) denote the stationary observed processes generated by \(\mathcal M\) and \(\mathcal M'\), respectively.  We first prove the forward direction.

Suppose
\[
    \mathcal M\sim_{\mathrm{sf}}\mathcal M'.
\]
By definition, there exists \(a>0\) such that
\[
    \{X_t^{\mathcal M'}\}_{t\in\mathbb Z}
    \overset{d}{=}
    \{\sqrt a\,X_t^{\mathcal M}\}_{t\in\mathbb Z}.
\]
Let
\[
    \Gamma_0=\operatorname{Cov}(X_t^{\mathcal M}),
    \qquad
    \Gamma_1=\operatorname{Cov}(X_t^{\mathcal M},X_{t-1}^{\mathcal M}),
\]
and define \(\Gamma_0'\) and \(\Gamma_1'\) analogously for \(\mathcal M'\).  Since multiplying the entire process by \(\sqrt a\) multiplies all covariance matrices by \(a\), we have
\[
    \Gamma_0'=a\Gamma_0,
    \qquad
    \Gamma_1'=a\Gamma_1.
\]
For a stationary VAR(1) process,
\[
    \Phi=\Gamma_1\Gamma_0^{-1},
    \qquad
    \Phi'=\Gamma_1'{\Gamma_0'}^{-1}.
\]
Therefore
\[
    \Phi'
    =
    (a\Gamma_1)(a\Gamma_0)^{-1}
    =
    \Gamma_1\Gamma_0^{-1}
    =
    \Phi.
\]
Moreover,
\[
    \Sigma_u
    =
    \Gamma_0-\Phi\Gamma_0\Phi^\top,
\]
and similarly
\[
    \Sigma_u'
    =
    \Gamma_0'-\Phi'\Gamma_0'{\Phi'}^\top.
\]
Using \(\Gamma_0'=a\Gamma_0\) and \(\Phi'=\Phi\), we obtain
\[
    \Sigma_u'
    =
    a\Gamma_0-\Phi(a\Gamma_0)\Phi^\top
    =
    a(\Gamma_0-\Phi\Gamma_0\Phi^\top)
    =
    a\Sigma_u.
\]
Thus \(\Phi'=\Phi\) and \(\Sigma_u'=a\Sigma_u\).

Conversely, suppose there exists \(a>0\) such that
\[
    \Phi'=\Phi,
    \qquad
    \Sigma_u'=a\Sigma_u.
\]
Let \(\{X_t^{\mathcal M}\}\) be the stationary reduced-form process
\[
    X_t^{\mathcal M}=\Phi X_{t-1}^{\mathcal M}+u_t,
    \qquad
    \operatorname{Cov}(u_t)=\Sigma_u.
\]
Define
\[
    Y_t:=\sqrt a\,X_t^{\mathcal M}.
\]
Then
\[
    Y_t
    =
    \sqrt a\,\Phi X_{t-1}^{\mathcal M}
    +
    \sqrt a\,u_t
    =
    \Phi Y_{t-1}
    +
    \widetilde u_t,
\]
where \(\widetilde u_t:=\sqrt a\,u_t\) has covariance
\[
    \operatorname{Cov}(\widetilde u_t)=a\Sigma_u=\Sigma_u'.
\]
Thus \(\{Y_t\}\) is a stationary Gaussian VAR(1) process with reduced-form parameters \((\Phi',\Sigma_u')\).  Since a stable Gaussian VAR(1) law is determined by its reduced-form parameters, \(\{Y_t\}\) has the same law as \(\{X_t^{\mathcal M'}\}\).  Hence
\[
    \{X_t^{\mathcal M'}\}_{t\in\mathbb Z}
    \overset{d}{=}
    \{\sqrt a\,X_t^{\mathcal M}\}_{t\in\mathbb Z},
\]
which means
\[
    \mathcal M\sim_{\mathrm{sf}}\mathcal M'.
\]
\end{proof}

\subsection{Proof of Theorem~\ref{thm:sf_orbit}}

\begin{proof}
First suppose
\[
    \mathcal M\sim_{\mathrm{sf}}\mathcal M'.
\]
By Proposition~\ref{prop:sf_reduced_form}, there exists \(a>0\) such that
\[
    \Phi'=\Phi,
    \qquad
    \Sigma_u'=a\Sigma_u.
\]
Using the structural forms of the reduced-form residual covariances,
\[
    \Sigma_u=\sigma^2B^{-1}B^{-\top},
    \qquad
    \Sigma_u'={\sigma'}^2{B'}^{-1}{B'}^{-\top},
\]
we get
\[
    {\sigma'}^2{B'}^{-1}{B'}^{-\top}
    =
    a\sigma^2B^{-1}B^{-\top}.
\]
Taking inverses gives
\[
    {\sigma'}^{-2}{B'}^\top B'
    =
    (a\sigma^2)^{-1}B^\top B.
\]
Equivalently,
\[
    {B'}^\top B'
    =
    \frac{{\sigma'}^2}{a\sigma^2}B^\top B.
\]
Define
\[
    c:=\frac{\sigma'}{\sqrt a\,\sigma}>0.
\]
Then
\[
    {B'}^\top B'=c^2B^\top B.
\]
By Lemma~\ref{lem:gram_orthogonal}, applied with \(C=B\), \(D=B'\), and \(\lambda=c^2\), there exists \(Q\in\mathcal O(p)\) such that
\[
    B'=cQB.
\]
Since \(\Phi'=\Phi\),
\[
    {B'}^{-1}A_1'=B^{-1}A_1.
\]
Multiplying by \(B'\) on the left and using \(B'=cQB\) gives
\[
    A_1'
    =
    B'B^{-1}A_1
    =
    cQA_1.
\]
Thus
\[
    B'=cQB,
    \qquad
    A_1'=cQA_1.
\]

Conversely, suppose there exist \(Q\in\mathcal O(p)\) and \(c>0\) such that
\[
    B'=cQB,
    \qquad
    A_1'=cQA_1.
\]
Then, as in the proof of Theorem~\ref{thm:obs_eq_orbit},
\[
    \Phi'
    =
    {B'}^{-1}A_1'
    =
    B^{-1}A_1
    =
    \Phi.
\]
For the residual covariance,
\[
\begin{aligned}
    \Sigma_u'
    &=
    {\sigma'}^2{B'}^{-1}{B'}^{-\top} \\
    &=
    {\sigma'}^2(cQB)^{-1}(cQB)^{-\top} \\
    &=
    \frac{{\sigma'}^2}{c^2}B^{-1}B^{-\top} \\
    &=
    \frac{{\sigma'}^2}{c^2\sigma^2}
    \Sigma_u.
\end{aligned}
\]
Therefore
\[
    \Sigma_u'=a\Sigma_u
    \qquad
    \text{with}
    \qquad
    a:=\frac{{\sigma'}^2}{c^2\sigma^2}>0.
\]
By Proposition~\ref{prop:sf_reduced_form},
\[
    \mathcal M\sim_{\mathrm{sf}}\mathcal M'.
\]
The displayed form of the scale-free observational equivalence class follows by writing \(A_0'=I-B'=I-cQB\) and by observing that, in the scale-free equivalence class, the noise scale may be any positive value, denoted by \(\tau>0\).
\end{proof}

\subsection{Proof of Proposition~\ref{prop:sf_orbit_closed_form}}

\begin{proof}
The scale-free observational alignment discrepancy is
\[
    \Delta_{\mathrm{align}}^{\mathrm{sf}}(\mathcal M'\mid\mathcal M)
    =
    \inf_{\substack{Q\in\mathcal O(p)\\ c>0}}
    \|S'-cQS\|_F^2.
\]
Expanding the squared Frobenius norm gives
\[
    \|S'-cQS\|_F^2
    =
    \|S'\|_F^2
    +
    c^2\|S\|_F^2
    -
    2c\,\operatorname{tr}(QS{S'}^\top).
\]
As in the proof of Proposition~\ref{prop:obs_orbit_closed_form}, let
\[
    C=S{S'}^\top
\]
and let
\[
    C=U\operatorname{diag}(\gamma_1,\ldots,\gamma_p)V^\top
\]
be a singular value decomposition.  By von Neumann's trace inequality,
\[
    \max_{Q\in\mathcal O(p)}
    \operatorname{tr}(QC)
    =
    \sum_{i=1}^p\gamma_i
    =
    \|C\|_*.
\]
Writing
\[
    \alpha:=\sum_{i=1}^p\gamma_i=\|S{S'}^\top\|_*,
\]
one maximizer is \(Q^\star=VU^\top\).  After optimizing over \(Q\), the remaining problem is
\[
    \inf_{c>0}
    \left\{
        \|S'\|_F^2
        +
        c^2\|S\|_F^2
        -
        2c\alpha
    \right\}.
\]
Since \(B\) is invertible, \(S\neq0\), so \(\|S\|_F^2>0\).

If \(\alpha>0\), the quadratic is minimized at
\[
    c^\star=\frac{\alpha}{\|S\|_F^2}>0.
\]
Substituting this value gives
\[
    \Delta_{\mathrm{align}}^{\mathrm{sf}}(\mathcal M'\mid\mathcal M)
    =
    \|S'\|_F^2
    -
    \frac{\alpha^2}{\|S\|_F^2}.
\]
If \(\alpha=0\), the objective after optimizing over \(Q\) is
\[
    \|S'\|_F^2+c^2\|S\|_F^2.
\]
Because the optimization is over \(c>0\), the infimum is approached as \(c\downarrow0\), and the infimum value is \(\|S'\|_F^2\), which is exactly the same value given by the displayed formula with \(\alpha=0\).  This proves the result.
\end{proof}

\subsection{Proof of Proposition~\ref{prop:canonical_empirical_orbit}}

\begin{proof}
We first note that \(\widehat\Sigma_u\succ0\) implies
\[
    \widehat\Omega_u=\widehat\Sigma_u^{-1}\succ0.
\]
Since
\[
    \widehat B_{\mathrm{can}}^\top\widehat B_{\mathrm{can}}
    =
    \widehat\Omega_u,
\]
the matrix \(\widehat B_{\mathrm{can}}\) is invertible.

We prove the first statement.  By definition,
\[
    \widehat\Gamma_{\mathrm{can}}
    =
    \widehat B_{\mathrm{can}}\widehat\Phi.
\]
Therefore,
\[
    \widehat B_{\mathrm{can}}^{-1}\widehat\Gamma_{\mathrm{can}}
    =
    \widehat B_{\mathrm{can}}^{-1}
    \widehat B_{\mathrm{can}}\widehat\Phi
    =
    \widehat\Phi.
\]
Moreover,
\[
\begin{aligned}
    \widehat B_{\mathrm{can}}^{-1}
    \widehat B_{\mathrm{can}}^{-\top}
    &=
    \left(
        \widehat B_{\mathrm{can}}^\top
        \widehat B_{\mathrm{can}}
    \right)^{-1} \\
    &=
    \widehat\Omega_u^{-1}
    =
    \widehat\Sigma_u.
\end{aligned}
\]
Thus the canonical unnormalized structural representation
\[
    \widehat{\mathcal M}_{\mathrm{can}}
    =
    (
        I-\widehat B_{\mathrm{can}},
        \widehat\Gamma_{\mathrm{can}},
        1
    )
\]
induces the empirical reduced-form parameters
\[
    (\widehat\Phi,\widehat\Sigma_u).
\]

We now prove the second statement.  Let
\[
    \widetilde{\mathcal M}
    =
    (\widetilde A_0,\widetilde A_1,\widetilde\sigma)
\]
be any structural representation with
\[
    \widetilde B:=I-\widetilde A_0
\]
invertible and \(\widetilde\sigma>0\).  Suppose first that \(\widetilde{\mathcal M}\) induces the same empirical reduced-form parameters:
\[
    \widetilde B^{-1}\widetilde A_1
    =
    \widehat\Phi,
    \qquad
    \widetilde\sigma^2
    \widetilde B^{-1}\widetilde B^{-\top}
    =
    \widehat\Sigma_u.
\]
Using the identity already proved for the canonical representative,
\[
    \widehat\Sigma_u
    =
    \widehat B_{\mathrm{can}}^{-1}
    \widehat B_{\mathrm{can}}^{-\top}.
\]
Hence
\[
    \widetilde\sigma^2
    \widetilde B^{-1}\widetilde B^{-\top}
    =
    \widehat B_{\mathrm{can}}^{-1}
    \widehat B_{\mathrm{can}}^{-\top}.
\]
Taking inverses gives
\[
    \widetilde\sigma^{-2}
    \widetilde B^\top\widetilde B
    =
    \widehat B_{\mathrm{can}}^\top
    \widehat B_{\mathrm{can}}.
\]
Equivalently,
\[
    \widetilde B^\top\widetilde B
    =
    \widetilde\sigma^2
    \widehat B_{\mathrm{can}}^\top
    \widehat B_{\mathrm{can}}.
\]
By Lemma~\ref{lem:gram_orthogonal}, applied with
\[
    C=\widehat B_{\mathrm{can}},
    \qquad
    D=\widetilde B,
    \qquad
    \lambda=\widetilde\sigma^2,
\]
there exists \(Q\in\mathcal O(p)\) such that
\[
    \widetilde B
    =
    \widetilde\sigma Q\widehat B_{\mathrm{can}}.
\]
Let
\[
    c:=\widetilde\sigma>0.
\]
Then
\[
    \widetilde B
    =
    cQ\widehat B_{\mathrm{can}}.
\]
Since
\[
    \widetilde B^{-1}\widetilde A_1
    =
    \widehat\Phi,
\]
we have
\[
    \widetilde A_1
    =
    \widetilde B\widehat\Phi
    =
    cQ\widehat B_{\mathrm{can}}\widehat\Phi
    =
    cQ\widehat\Gamma_{\mathrm{can}}.
\]
Also, by definition of \(c\),
\[
    \widetilde\sigma=c.
\]
This proves the forward direction.

Conversely, suppose there exist \(Q\in\mathcal O(p)\) and \(c>0\) such that
\[
    \widetilde B
    =
    cQ\widehat B_{\mathrm{can}},
    \qquad
    \widetilde A_1
    =
    cQ\widehat\Gamma_{\mathrm{can}},
    \qquad
    \widetilde\sigma
    =
    c.
\]
Then
\[
\begin{aligned}
    \widetilde B^{-1}\widetilde A_1
    &=
    (cQ\widehat B_{\mathrm{can}})^{-1}
    (cQ\widehat\Gamma_{\mathrm{can}}) \\
    &=
    \widehat B_{\mathrm{can}}^{-1}Q^\top Q
    \widehat\Gamma_{\mathrm{can}} \\
    &=
    \widehat B_{\mathrm{can}}^{-1}
    \widehat\Gamma_{\mathrm{can}}
    =
    \widehat\Phi.
\end{aligned}
\]
Similarly,
\[
\begin{aligned}
    \widetilde\sigma^2
    \widetilde B^{-1}\widetilde B^{-\top}
    &=
    c^2
    (cQ\widehat B_{\mathrm{can}})^{-1}
    (cQ\widehat B_{\mathrm{can}})^{-\top} \\
    &=
    c^2
    \left(c^{-1}\widehat B_{\mathrm{can}}^{-1}Q^\top\right)
    \left(c^{-1}Q\widehat B_{\mathrm{can}}^{-\top}\right) \\
    &=
    \widehat B_{\mathrm{can}}^{-1}Q^\top Q
    \widehat B_{\mathrm{can}}^{-\top} \\
    &=
    \widehat B_{\mathrm{can}}^{-1}
    \widehat B_{\mathrm{can}}^{-\top}
    =
    \widehat\Sigma_u.
\end{aligned}
\]
Therefore \(\widetilde{\mathcal M}\) induces the same empirical reduced-form parameters \((\widehat\Phi,\widehat\Sigma_u)\).  This proves the proposition.
\end{proof}

\section{Experimental Setting and Details}\label{sec:exp-details}

\subsection{Synthetic Data Generation}
To evaluate the performance of the causal discovery algorithms, we simulated continuous time-series data following a first-order Vector Autoregressive, or VAR(1), process. We evaluated the models across varying graph dimensions of size $p \in \{5, 10, 15, 25, 50, 75, 100\}$. For each graph, we generated $T = 1000$ samples. 

The structural matrices defining the instantaneous effects ($A_0$) and lagged effects ($A_1$) were generated by initially applying a sparsity mask with an expected edge density of 30\%. The non-zero entries were drawn uniformly from $[-1.0, 1.0]$. To ensure admissibility,  the diagonal of $A_0$ was set to zero. Further,  we constrained the spectral radius of both $A_0$ and the reduced-form transition matrix $\Phi = (I - A_0)^{-1}A_1$ such that if either spectral radius exceeded a threshold of 0.85, the respective matrices were scaled down (multiplied by a constant) to satisfy this stability threshold.

The exogenous noise terms, $e_t$, were simulated as independent Gaussian variables. To test the algorithms under varying degrees of heteroscedasticity, the base standard deviation for the noise across nodes was centered at $\sigma_{\mathrm{nom}} = 1.0$, with individual node noise standard deviations drawn from a Gaussian distribution with standard deviation $\sigma_{\mathrm{std}} \in \{0.00, 0.025, 0.075, 0.10, 0.15\}$. The final observations $X_t$ were generated recursively via
$$X_t = (I - A_0)^{-1}(A_1 X_{t-1} + e_t)$$
and the whole process was repeated for 5 episodes.

\subsection{Baseline Causal Discovery Methods}

We compared our proposed approach against three established structure learning algorithms designed to predict the lag matrices $\widehat{A}_0$ and $\widehat{A}_1$.

\textbf{Equal Variance Greedy DAG Search (EqVarGDS):} A baseline approach based on an extension of the cross-sectional algorithm in~\citep{peters2014identifiability} that sequentially orders nodes based on conditional variance. Nodes are ordered iteratively by minimizing their conditional variance given previously selected nodes. The instantaneous adjacency matrix $\widehat{A}_0$ is estimated via least-squares regression according to this inferred topological ordering, utilizing an alpha-level threshold ($\alpha = 0.05$) to prune insignificant edges. Subsequently, the lagged adjacency matrix is derived algebraically using the empirical reduced-form transition matrix $\Phi_{\text{emp}}$ via the relation $\widehat{A}_1 = (I - \widehat{A}_0)\Phi_{\text{emp}}$.

\textbf{VAR-LiNGAM~\citep{hyvarinen2010estimation}:} This method first regresses $X_t$ on $X_{t-1}$ to extract the reduced-form residuals $\hat{u}_t$. The DirectLiNGAM algorithm (or HighDimDirectLiNGAM when $T \leq p$) is then applied to the residuals to estimate the instantaneous causal structure. Edges are subsequently pruned based on statistical significance testing ($\alpha = 0.05$).

\textbf{DYNOTEARS~\citep{pamfil2020dynotears}:} A continuous optimization framework for dynamic networks. The hyperparameters controlling the $\ell_1$ penalty on the instantaneous and lagged structures were both set to $\lambda_w = 0.05$ and $\lambda_a = 0.05$, respectively, with a strict zero-thresholding bound of $0.0$.

\subsection{Proposed Method: ENVAR Optimization}
The ENVAR algorithm formulates the causal discovery problem as an optimization over the space of orthogonal matrices to translate the empirical reduced-form parameters ($\Phi_{\text{emp}}$ and $\Sigma_{u,\text{emp}}$) into structural parameters $A_0$ and $A_1$. 

To satisfy the strict requirement of an orthogonal transformation matrix $Q$, we employ PyTorch's orthogonal parameterization framework, ensuring $Q^T Q = I$ at every optimization step without relying on soft penalty terms. This approach leverages automatic differentiation (Autograd) to efficiently navigate the complex parameter space. In other words, rather than relying on soft regularization penalties, that often yield unstable approximations, we employed PyTorch's structural orthogonal parameterization (\texttt{torch.nn.utils.parametrizations.orthogonal}). This mechanism applies a mathematical projection to an underlying unconstrained matrix during the forward pass. Consequently, it guarantees that the $Q$ matrix evaluated in the loss function is strictly orthogonal at every step, allowing the optimizer to search exclusively along the orthogonal manifold using standard gradient-based updates. Additionally, we optimize a strictly positive scaling factor $c$. To enforce the positivity constraint for $c$, we parameterized the variable in the logarithmic domain, by defining the learnable parameter as $\log c$ and evaluating $c = \exp(\log c)$ during the forward pass. Finally, the structural matrices are reconstructed as
$$\widehat{A}_0 = I - c \, Q L_{0,\text{emp}}$$
$$\widehat{A}_1 = c \, Q \Gamma_{0,\text{emp}}$$
where $L_{0,\text{emp}}$ and $\Gamma_{0,\text{emp}}$ are derived from the Cholesky decomposition of the inverse empirical residual covariance and its relation to the empirical transition matrix.

The objective function minimizes a weighted sum of structural sparsity, graph hollowness, and reconstruction error. During our experiments, the reconstruction weight was set to $w_{\text{recons}} = 0.0$. The edge sparsity weights were fixed to $w_1 = 1.0$ (for $\widehat{A}_0$) and $w_2 = 1.0$ (for $\widehat{A}_1$). The hollowness penalty $w_{\text{hollow}}$, which forces the diagonal of $\widehat{A}_0$ toward zero, was scaled based on the dimensionality of the graph: $7.5$ for $p \le 25$, $5.0$ for $25 < p \le 75$, and $2.5$ for $p > 75$. Normalization constants ($\text{norm}_{A0}, \text{norm}_{A1}, \text{norm}_{\text{hollow}}$) were pre-calculated from a baseline random orthogonal projection to ensure stable gradients across dimensions. The model parameters were optimized using the Adam optimizer and an adaptive learning rate defined as $\eta = 5 \times 10^{-3} \cdot (5 / p)$. Gradients were clipped at a maximum norm of 1.0, running for up to 10,000 steps for larger graphs ($p > 10$).

\subsection{Evaluation Metric}

To assess the accuracy of the recovered graphs relative to the ground truth, we utilized the Scale-Free Observational Alignment Discrepancy (Scale-Free OAD).

\subsection{fMRI Data Preprocessing}

For the real fMRI analysis, we used minimally preprocessed
\citep{glasser2013minimal} task-fMRI data from the motor task of the
Human Connectome Project (HCP) S1200 release \citep{van2013wu, barch2013function}. In the HCP motor task, participants are presented with
visual cues instructing them to tap their left or right fingers,
squeeze their left or right toes, or move their tongue, in 12-second
blocks each preceded by a 3-second cue, adapted from the paradigm of
\citet{barch2013function}. Task-fMRI images were collected with the same EPI pulse sequence as
resting-state fMRI in HCP, with the following parameters:
TR $= 720$\,ms, TE $= 33.1$\,ms, flip angle $= 52^\circ$,
FOV $= 208 \times 180$\,mm, matrix $= 104 \times 90$, slice thickness
$= 2.0$\,mm, number of slices $= 72$ (2.0\,mm isotropic), multiband
factor $= 8$, and echo spacing $= 0.58$\,ms. Each motor-task run consisted of $284$
frames (approximately 3 minutes and 24 seconds of scan time); each subject completed two runs (left-right and right-left
phase encoding), and we used both runs in our analysis. Brains were
normalized to fslr32k via the multimodal surface matching
(MSM)-All registration. 

We applied the following preprocessing steps to the minimally preprocessed time-series. First, we computed the global signal defined as the  mean across all the voxels at each time point, and regressed it out from each voxel \citep{murphy2017towards}. We then applied linear detrending to remove low-frequency scanner drift. Voxels were parcellated into $116$ regions of interest (ROI) using Schaefer-100 7-network cortical
parcellation \citep{schaefer2018local} together with the Tian S1 subcortical atlas (16 regions) 
\citep{tian2020topographic}. Each ROI time-series was then $z$-scored across time. The HCP experiments were carried out by the WU-Minn consortium, and
its adherence to ethical standards was approved by the internal
review boards of the respective institutions. Explicit informed
consent was acquired from all participants \citep{van2013wu}.

\subsection{Assessment of Causal Graphs from fMRI Data}

To validate the biological relevance of the estimated causal structures from real functional data, we employed the Neurosynth meta-analytic database~\citep{yarkoni2011large} to decode the cognitive terms associated with any given vector of nodal centralities (in-degree, out-degree, or causal flow. After we applied the ENVAR method to HCP motor task data of each subject (combined across both runs), the obtained the structural matrices $A_0$ and $A_1$ were binarized using a cumulative effect-size thresholding approach to isolate the most robust connections and suppress noise. Specifically, for both $A_0$ and $A_1$, we ranked the absolute values of the estimated edges and retained only the strongest connections that cumulatively accounted for 85\% of the total absolute edge weight in the respective matrix. The final binary causal graph for each run was then constructed as the logical union of the binarized $A_0$ and $A_1$ matrices. To map these network topologies to distinct cognitive functions, we extracted the network in-degree, out-degree, and causal-flow for each ROI from the final binary graphs. Each of these nodal centrality vectors (for each subject) were then provied as input to Neurosynth, which then matched them in MNI coordinates (of atlas ROI centroids) against reported study activations in the literature. This matching utilized a Gaussian spatial kernel ($\sigma = 6$ mm) bounded by a 12 mm search radius. The ROI in-degrees, out-degrees, and net-flows were used to weight the spatially associated studies, which were subsequently multiplied by the Neurosynth TF-IDF feature matrix to generate a continuous score for each cognitive term. Finally, these decoded term scores were averaged across all subjects and runs to construct a population-level cognitive profile for each graph discovery method, illustrated as the word clouds in~\cref{fig:supp_hcp_task_motor_wc}.

\subsection{Computational Resources}

Experiments were performed partly on a local workstation (Lenovo P620 with AMD 3970X 32-Core processor, Nvidia GeForce RTX 2080 GPU, and 512GB of RAM) and partly on a High-Performance Computing cluster (typically with 8 CPU cores and 32GB of RAM per job).

\setcounter{figure}{0}
\renewcommand{\thefigure}{\thesection.\arabic{figure}}

\newpage

\section{Supplementary Figures}\label{sec:supp-figs}

\begin{figure}[h!]
  \centering
  \begin{subfigure}[b]{0.40\textwidth}
    \centering
    \includegraphics[height=4.00cm]{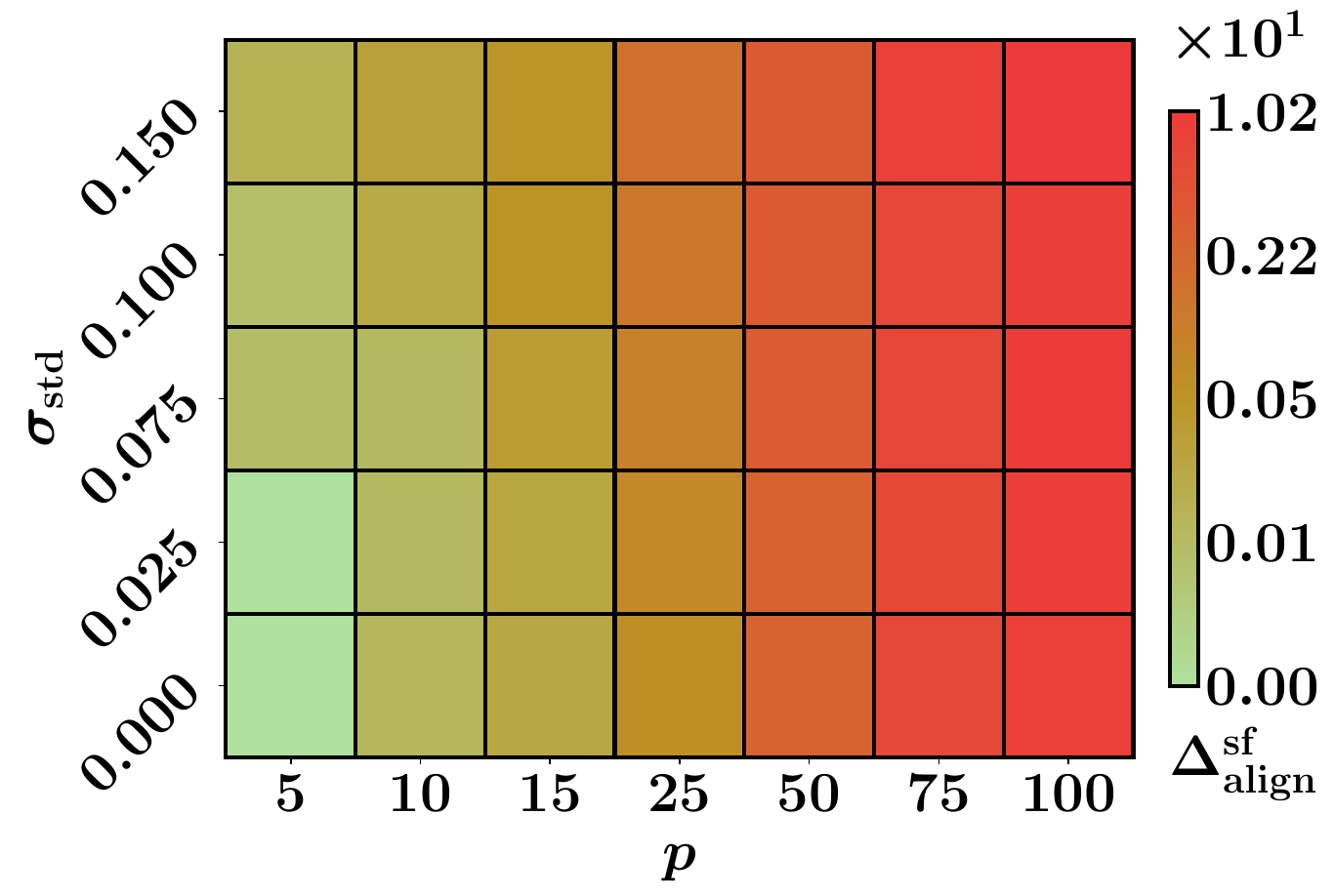}
    \caption{scale-free alignment discrepancy}
    \label{fig:supp_obs_orb_disc_eq_var_violation_ENVAR}
  \end{subfigure}
  \hfill
  \begin{subfigure}[b]{0.55\textwidth}
    \centering
    \includegraphics[height=5.00cm]{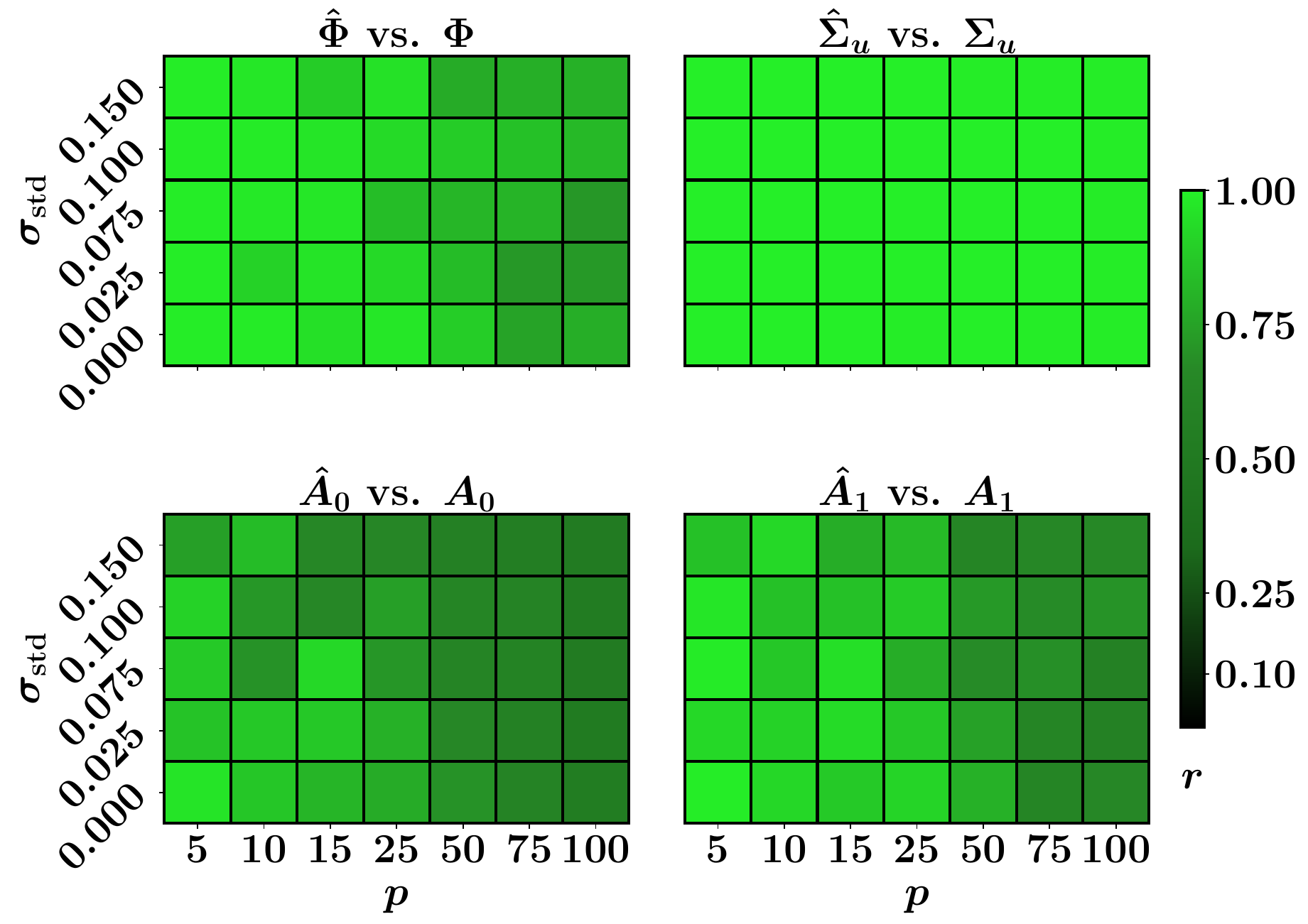}
    \caption{Pearson correlation coefficient ($r$)}
    \label{fig:supp_corr_eq_var_violation_ENVAR}
  \end{subfigure}
  \caption{Performance of ENVAR vs. the number of nodes, under the violation of the equal variance assumption for the noise, i.e., $e_t \stackrel{\mathrm{i.i.d.}}{\sim} \mathcal{N}(0, \sigma_i^2), \sigma_i \sim \mathcal{N}(1.0, \sigma_{\mathrm{std}}^2)$.  The number of samples is considered as $T = 1000$. The values are mean across 5 episodes. \textbf{(a)} The scale-free alignment discrepancy of the predicted lag matrices to the ground truth equivalence class. \textbf{(b)} The Pearson correlation coefficient ($r$) of the predicted reduced form parameters (only significant correlations with $p<0.05$ are considered) with the corresponding ground truth reduced form parameters.}
  \label{fig:supp_eq_var_violation_ENVAR}
\end{figure}

\begin{figure}[h]
  \centering
  \begin{subfigure}[b]{0.40\textwidth}
    \centering
    \includegraphics[height=4.00cm]{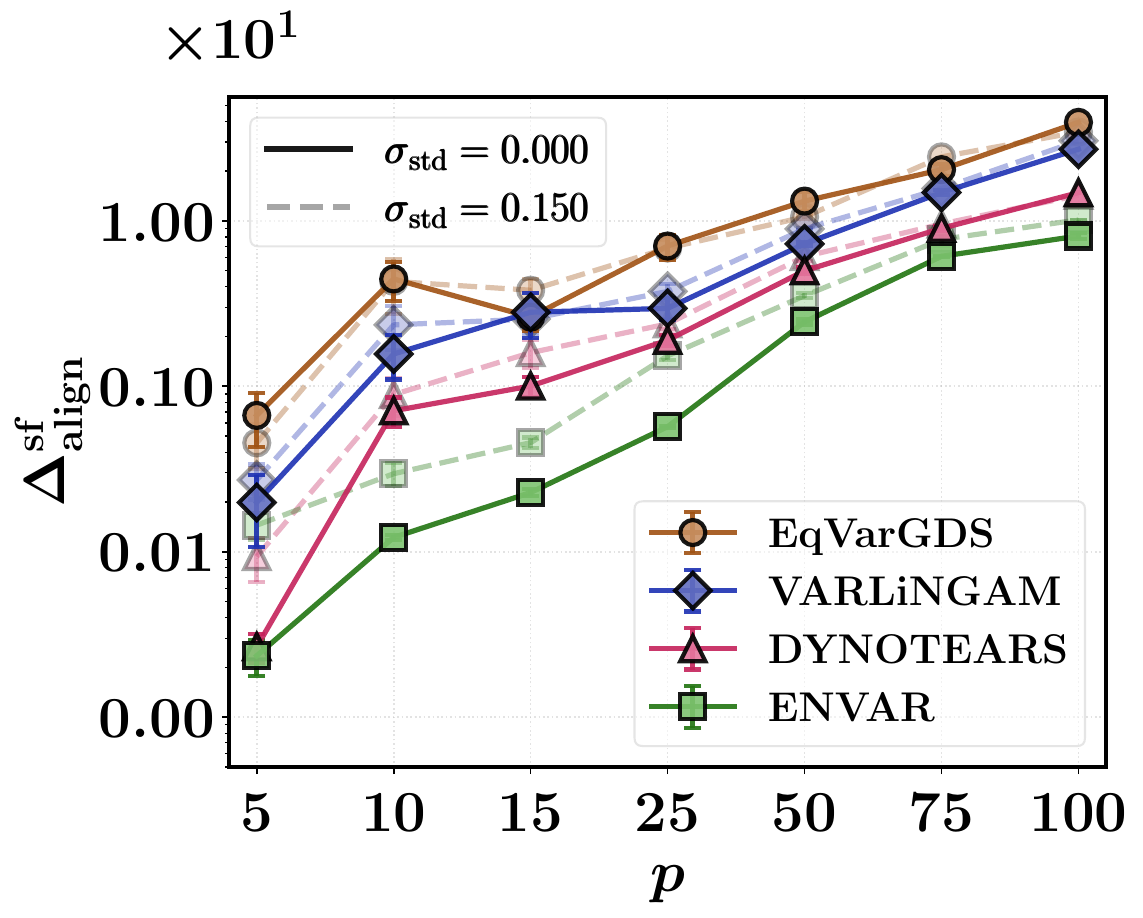}
    \caption{scale-free alignment discrepancy}
    \label{fig:supp_obs_orb_disc_eq_var_violation}
  \end{subfigure}
  \hfill
  \begin{subfigure}[b]{0.55\textwidth}
    \centering
    \includegraphics[height=5.00cm]{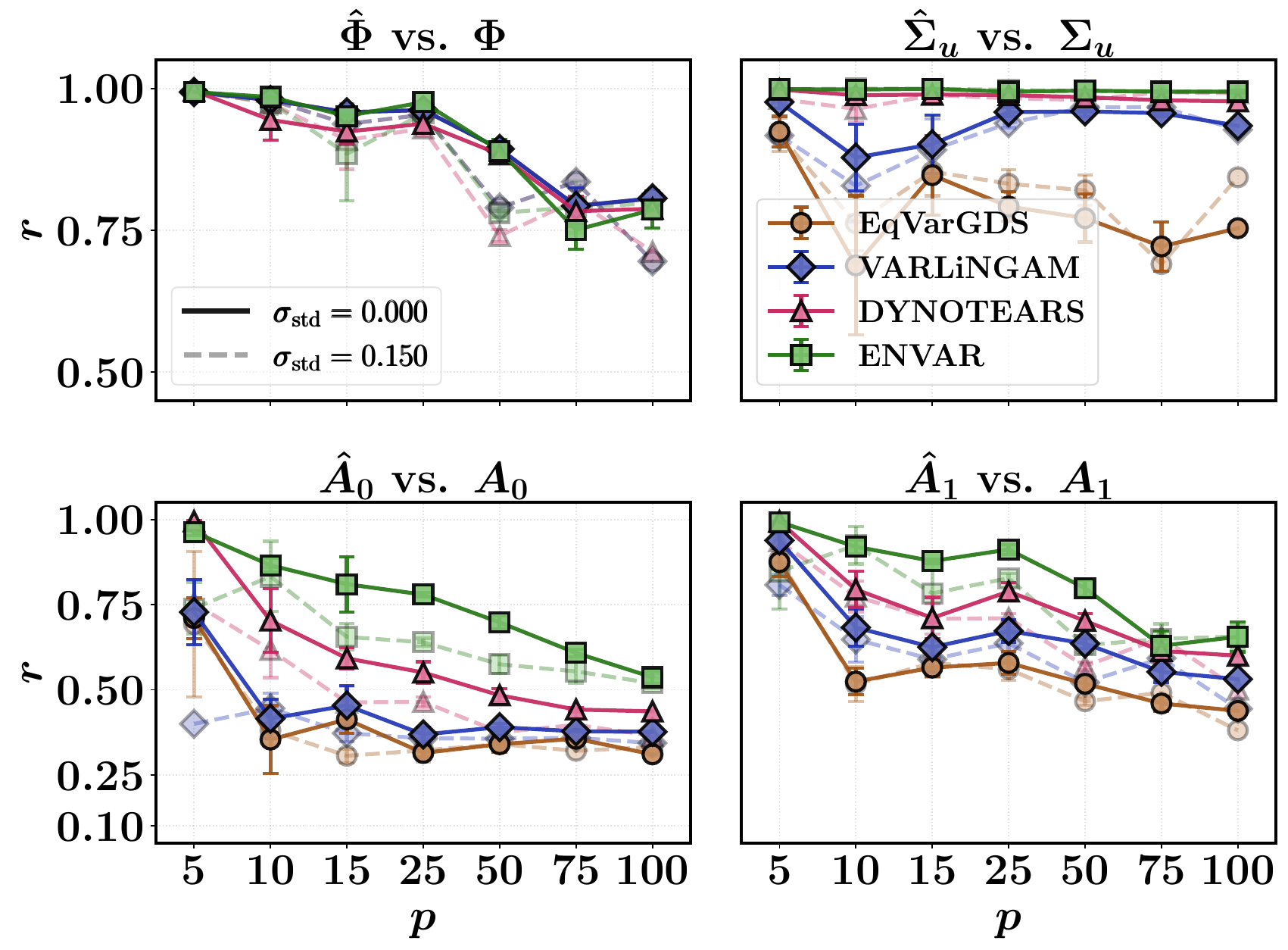}
    \caption{Pearson correlation coefficient ($r$)}
    \label{fig:supp_corr_eq_var_violation}
  \end{subfigure}
  \caption{Performance comparison of EqVarGDS, ENVAR, VARLiNGAM, and DYNOTEARS vs. the number of nodes, under the violation of the equal variance assumption for the noise, i.e., $e_{t,i} \stackrel{\mathrm{i.i.d.}}{\sim} \mathcal{N}(0, \sigma_i^2), \sigma_i \sim \mathcal{N}(1.0, \sigma_{\mathrm{std}}^2)$. The number of samples is considered $T = 1000$. The error bars represent SEM across 5 episodes. \textbf{(a)} The scale-free alignment discrepancy of the predicted lag matrices to the ground truth equivalence class. Note the logarithmic scale on the ordinate. \textbf{(b)} The Pearson correlation coefficient ($r$) between the predicted reduced form parameters and their corresponding ground truths. Only significant correlations with p-value $<0.05$ are considered for better comparison. EqVarGDS and VARLiNGAM lines overlap with ENVAR in the top panel.}
\end{figure}

\begin{figure}[htbp]
    \centering
    \begin{subfigure}{0.32\textwidth}
        \centering
        \includegraphics[trim=0cm 4cm 0cm 4cm, clip, width=\linewidth]{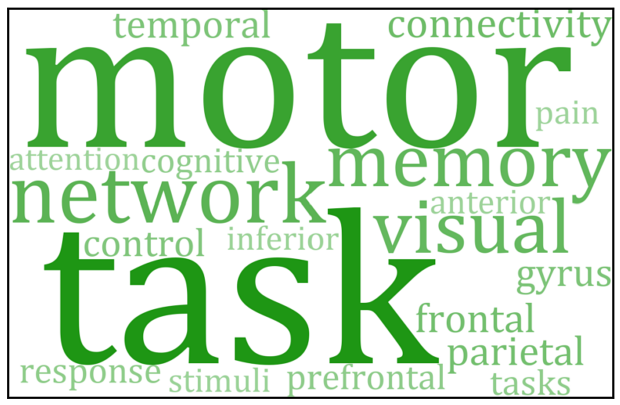}
        \caption{In-degree}
        \label{fig:supp_hcp_task_motor_wc_indeg}
    \end{subfigure}
    \hfill
    \begin{subfigure}{0.32\textwidth}
        \centering
        \includegraphics[trim=0cm 4cm 0cm 4cm, clip, width=\linewidth]{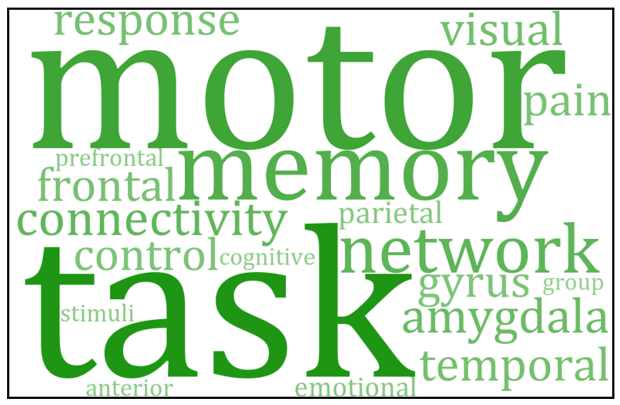}
        \caption{Out-degree}
        \label{fig:supp_hcp_task_motor_wc_outdeg}
    \end{subfigure}
    \hfill
    \begin{subfigure}{0.32\textwidth}
        \centering
        \includegraphics[trim=0cm 4cm 0cm 4cm, clip, width=\linewidth]{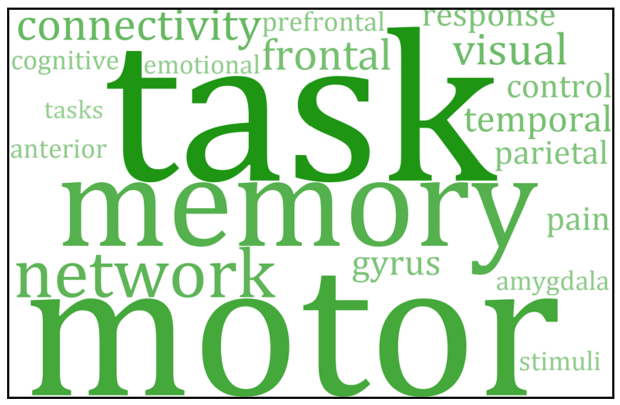}
        \caption{Net-Flow}
        \label{fig:supp_hcp_task_motor_wc_netflow}
    \end{subfigure}
    \caption{Word clouds mapping the in-degree, out-degree, and net flow of the predicted final binary graphs based on HCP motor task data to cognitive terms using Neurosynth~\citep{yarkoni2011large}.}
    \label{fig:supp_hcp_task_motor_wc}
\end{figure}

\end{document}